\documentclass[preprint,10pt]{elsarticle}

\usepackage{lineno,hyperref}
\usepackage{amsmath}
\usepackage{subfigure}
\usepackage{caption}
\usepackage{xcolor}
\usepackage{amssymb}
\usepackage{lipsum}
\usepackage{amsfonts}
\usepackage{graphicx}
\usepackage{epstopdf}
\usepackage{algorithmic}
\usepackage{amsmath}
\usepackage{bm}
\usepackage{mathrsfs}
\usepackage{amsthm}
\usepackage{soul}
\usepackage{algorithm}
\usepackage{algorithmic}
  
\newtheorem{theorem}{Theorem}
\newtheorem{lemma}{Lemma}

\newtheorem{remark}{Remark}

\newtheorem{problem}[theorem]{Problem} 
\newtheorem{proposition}{Proposition}  
\modulolinenumbers[5]

\journal{Journal of Fluids and Structures}

\biboptions{numbers,sort&compress}

\begin{document}

\begin{frontmatter}

\title{An energy stable one-field monolithic arbitrary Lagrangian-Eulerian formulation for fluid-structure interaction} 


\author[add0]{Yongxing Wang\corref{mycorrespondingauthor}}
\cortext[mycorrespondingauthor]{Corresponding author}
\ead{scsywan@leeds.ac.uk/yongxingwang6@gmail.com}

\author[add1]{Peter K. Jimack}
\author[add1]{Mark A. Walkley}
\author[add2]{Olivier Pironneau}

\address[add0]{School of  Mechanical Engineering, University of Leeds, Leeds, UK, LS2 9JT}
\address[add1]{School of Computing, University of Leeds, Leeds, UK, LS2 9JT}
\address[add2]{Laboratoire Jacques-Louis Lions, Sorbonne Universit$\acute{e}$s, Paris, France}

\begin{abstract}
In this article we present a one-field monolithic finite element method in the Arbitrary Lagrangian-Eulerian (ALE) formulation for Fluid-Structure Interaction (FSI) problems. The method only solves for one velocity field in the whole FSI domain, and it solves in a monolithic manner so that the fluid solid interface conditions are satisfied automatically. We prove that the proposed scheme is unconditionally stable, through energy analysis, by utilising a conservative formulation and an exact quadrature rule. We implement the algorithm using both ${\bf F}$-scheme and ${\bf d}$-scheme, and demonstrate that the former has the same formulation in two and three dimensions. Finally several numerical examples are presented to validate this methodology, including combination with remesh techniques to handle the case of very large solid displacement.
\end{abstract}

\begin{keyword}
fluid structure interaction \sep finite element \sep one field \sep monolithic scheme \sep arbitrary Lagrangian-Eulerian \sep energy stable
\end{keyword}

\end{frontmatter}


\section{Introduction}
\label{sec_introduction}

Numerical methods for Fluid-Structure Interaction (FSI) have been widely studied during the past decades, and a variety of methodologies have been developed in order to address different aspects of the FSI problem. However stability analyses of the existing numerical methods are rare especially when large solid deformation is involved. This paper is dedicated to establishing a robust stability analysis of a one-field monolithic FSI scheme in the Arbitrary Lagrangian-Eulerian (ALE) framework.

Monolithic methods have been regarded as the most robust FSI algorithms in the literature \cite{Heil_2004,Heil_2008,Muddle_2012,Hecht_2017,Wang_2017,Wang_2019,hubner2004monolithic}, which solve for the fluid and solid variables simultaneously in one equation system. Among these methodologies for FSI problems, the one-field approaches \cite{Hecht_2017,Wang_2017} express the solid equation in terms of velocity, thus only solve for one velocity in the whole FSI domain. In this case the whole system can be solved similarly to a modified fluid problem, and the coupling conditions at fluid and solid interface are automatically satisfied.

The stability analysis when using the ALE framework is challenging, even for the pure fluid problem, due to the arbitrary moving frame \cite{nobile1999stability,formaggia2004stability,bonito2013time}. \cite{Boffi_2016,Boffi_2015} present an energy stable Fictitious Domain Method with Distributed Lagrangian Multiplier (FDM/DLM), and \cite{Hecht_2017,Pironneau_2016} present an energy stable Eulerian formulation by remeshing. In a previous study \cite{Wang_2019} we analysed the energy stability for a one-field FDM method. In this article we extend this one-field idea to the ALE formulation, and the stability result is achieved by expressing the fluid and solid equations in a conservative formulation. In this sense, the formulation is similar to the one introduced in \cite{Hecht_2017}. However it differs from \cite{Hecht_2017} in the following perspectives: (1) we formulate the solid in the reference domain and analyse in an ALE frame of reference, in which case the formulation and analysis are exactly the same for two and three dimensional cases, whereas \cite{Hecht_2017} formulates and analyses everything in the current domain, for which the three dimensional case is significantly more complicated \cite{chiang2017numerical}; (2) we update the solid deformation tensor (the ${\bf F}$-scheme) while \cite{Hecht_2017} updates the solid displacement (the ${\bf d}$-scheme); (3) we implement the scheme by solving an additional solid-like equation at each time step in order to move the mesh, whilst \cite{Hecht_2017} implements their scheme by remeshing which is expensive in the three dimensional case.

The paper is organized as follows. In Section \ref{sec_ale} the control equations for the FSI problem are introduced in an ALE framework. In Section \ref{weak_formulation} the finite element weak formulation is introduced, followed by spatial and time discretisations in Section \ref{section_discretization}. The main results of energy stability are presented in Section \ref{stability_energy}. Implementation details are considered in Section \ref{implementation} and numerical examples are given in Section \ref{sec_numerical_exs}, with some conclusions in Section \ref{sec_conclusion}.

\section{The arbitrary Lagrangian-Eulerian description for the FSI problem}
\label{sec_ale}
Let $\Omega_t^f\subset\mathbb{R}^d$ and $\Omega_t^s\subset\mathbb{R}^d$ be the fluid and solid domain respectively (which are time dependent regions), $\Gamma_t=\overline{\Omega}_t^f \cap \overline{\Omega}_t^s$ is the moving interface between the fluid and solid, and $\Omega_t=\overline{\Omega}_t^f \cup \overline{\Omega}_t^s$ has an outer boundary $\partial\Omega_t$, which can be fixed or moving as shown in Figure \ref{aleformulation}. The Eulerian description is convenient when we observe a fluid from a fixed frame, while the Lagrangian description is convenient when we observe a solid from a frame moving with it. An ALE frame of reference can be adopted when a fluid and solid share an interface and interact with each other as shown in Figure \ref{aleformulation}, in which case the frame moves arbitrarily from a reference configuration $\Omega_{t_0}$, chosen to be the same as the initial configuration at $t_0$, to a current configuration $\Omega_t$. Let us define a family of mappings $\mathcal{A}_t$:
\begin{equation}
\mathcal{A}_t: \Omega_{t_0} \subset \mathbb{R}^d\rightarrow \Omega_t \subset \mathbb{R}^d,
\end{equation}
with $d=2,3$ being the dimensions. We assume that $\mathcal{A}_t \in C^0\left(\overline{\Omega}_{t_0}\right)^d$ is one-to-one and invertible with continuous inverse $\mathcal{A}^{-1}_t \in C^0\left(\overline{\Omega}_t\right)^d$. Hence a point $\hat{\bf x} \in \Omega_{t_0}$ has a unique image ${\bf x} \in \Omega_t$ at time $t$, i.e. 
\begin{equation}\label{ale_mapping}
{\bf x}={\mathcal{A}}\left(\hat{{\bf x}},t\right)
=\mathcal{A}_t\left({\hat{\bf x}}\right),
\end{equation}
and a point ${\bf x} \in \Omega_t$ at time $t$ has a unique inverse image $\hat{\bf x} \in \Omega_{t_0}$
\begin{equation}
\hat{\bf x}=\hat{\mathcal{A}}\left({\bf x},t\right)
=\mathcal{A}^{-1}_t\left({\bf x}\right).
\end{equation}
We call ${\bf x}\in \Omega_t$ the Eulerian coordinate, and call its inverse image ${\hat{\bf x}}$, via the above mapping $\mathcal{A}^{-1}_t$, the ALE coordinate.  We assume that ${\mathcal{A}}\left(\hat{{\bf x}},t\right)$ is differentiable with respect to $t$ for all $\hat{{\bf x}}\in \Omega_{t_0}$, and define the velocity of the ALE frame as
\begin{equation}\label{definition_w}
{\bf w}\left({\hat{\bf x}},t\right)=\frac{\partial{\mathcal{A}}}{\partial t}\left({\hat{\bf x}},t\right).
\end{equation}
Given an Eulerian coordinate ${\bf x} \in \Omega_t$, its corresponding ALE coordinate $\hat{\bf x}_1 \in \Omega_{t_0}$ should be distinguished from its material (or Lagrangian) coordinate $\hat{\bf x}_2 \in \Omega_{t_0}$ as shown in Figure \ref{aleformulation}. In fact $\hat{\bf x}_2 \in \Omega_{t_0}$ (not necessarily the same as $\hat{\bf x}_1$) maps to ${\bf x} \in \Omega_t$ via the Lagrangian mapping, i.e., the trajectory of a material particle at $\hat{\bf x}_2$:
\begin{equation}\label{lagrangian_mapping}
\mathcal{F}_t: \hat{{\bf x}} \mapsto 
{\bf x}=\mathcal{F}\left({\hat{\bf x}},t\right),
\end{equation}
and the velocity of the material particle at $\hat{\bf x} \in \Omega_{t_0}$ is defined by
\begin{equation}
{\bf u}\left({\hat{\bf x}},t\right)=\frac{\partial{\mathcal{F}}}{\partial t}.
\end{equation}
\begin{figure}[h!]
\centering
\includegraphics[width=3.5in,angle=0]{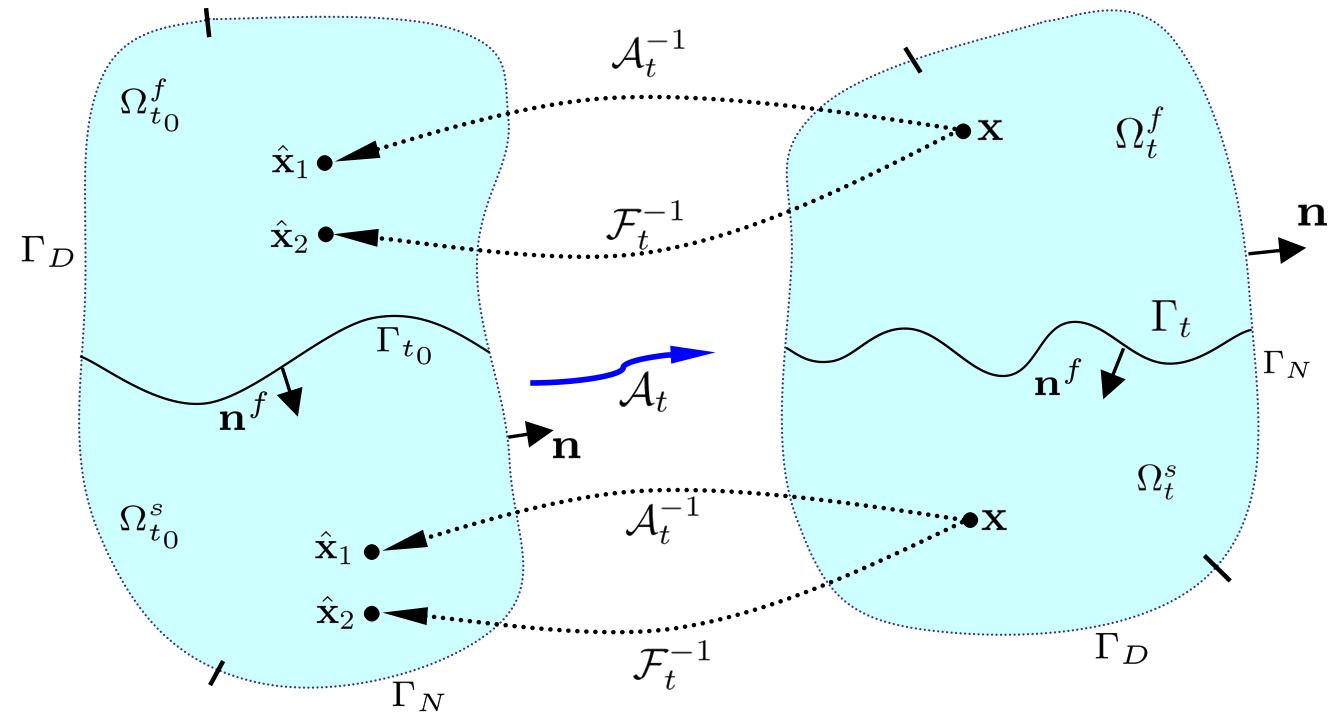}
\caption {\scriptsize ALE mapping from $\Omega_{t_0}$ to $\Omega_t$. Also shows the comparison between ALE mapping and Lagrangian mapping with Eulerian coordinate ${\bf x}$, ALE coordinate $\hat{\bf x}_1$ and material (Lagrangian) coordinate $\hat{\bf x}_2$. $\Gamma_t=\overline{\Omega}_t^f \cap \overline{\Omega}_t^s$ and $\Omega_t=\overline{\Omega}_t^f \cup \overline{\Omega}_t^s$, $\partial\Omega_t=\Gamma_D\cup\Gamma_N$.} 
\label{aleformulation}
\end{figure}

\begin{remark}
Although the Lagrangian configuration and the ALE configuration are not generally the same, both are chosen to have the initial configuration $\Omega_{t_0}$ in this article. We shall also construct the ALE mapping such that $\mathcal{A}_t\left(\Omega_{t_0}\right)$ coincides with $\mathcal{F}_t\left(\Omega_{t_0}\right)$ at all boundaries including the fluid-solid interface: $\mathcal{A}_t\left(\partial{\Omega_{t_0}}\right)=\mathcal{F}_t\left(\partial{\Omega_{t_0}}\right)$ and $\mathcal{A}_t\left(\partial{\Gamma_{t_0}}\right)=\mathcal{F}_t\left(\partial{\Gamma_{t_0}}\right)$.
\end{remark}
\begin{remark}
The ALE mapping is the mapping that is actually used to move the domain in this article, and the purpose of introducing the Lagrangian mapping is to discuss its related variables, such as particle velocity ${\bf u}$ and solid deformation tensor ${\bf F}$, which will be defined in the following context.
\end{remark}

Formulated in the current configuration, the conservation of momentum takes the same form in the fluid and solid:
\begin{equation} \label{momentum_equation}
\rho\frac{d{\bf u}\left({\bf x},t\right)}{dt}
={\rm div}\left({\bm\sigma}\right)+\rho{\bf g},
\end{equation}
with $\rho$, ${\bf g}$, ${\bf u}$ and ${\bm\sigma}$ being the density, gravity acceleration, velocity and Cauchy stress tensor respectively. Here we use the notation $
{\rho}=\left \{ 
\begin{matrix}
{{\rho}^f \quad in \quad \Omega_t^f} \\
{{\rho}^s \quad in \quad \Omega_t^s} \\
\end{matrix}\right.
$, with the superscript $f$ and $s$ denote fluid and solid respectively, and similar notations are also applied to ${\bf u}$ and $\bm\sigma$. In the above, $\frac{d({\cdot})}{dt}$ is the total derivative computed along the trajectory of a material particle at ${\bf x}$, i.e. via the Lagrangian mapping:
\begin{equation}
\frac{d{\bf u}\left({\bf x},t\right)}{dt}
=\frac{d{\bf u}\left(\mathcal{F}_t\left(\hat{\bf x}\right),t\right)}{dt}
=\left.\frac{\partial{\bf u}}{\partial t}\right|_{{\bf x}=\mathcal{F}\left(\hat{\bf x},t\right)}
+\left({\bf u}\cdot\nabla\right){\bf u}.
\end{equation}
Replacing the above partial time derivative by the total derivative of
\begin{equation}
\frac{d{\bf u}\left(\mathcal{A}_t\left(\hat{\bf x}\right),t\right)}{dt}
=\left.\frac{\partial {\bf u}}{\partial t}\right|_{{\bf x}=\mathcal{A}\left(\hat{\bf x},t\right)}
+\left({\bf w}\cdot\nabla\right){\bf u}
\end{equation}
leads to the ALE formulation of (\ref{momentum_equation})
\begin{equation} \label{momentum_equation_ale}
\rho\frac{d{\bf u}\left(\mathcal{A}_t\left(\hat{\bf x}\right),t\right)}{dt}
+\rho\left(\left({\bf u}-{\bf w}\right)\cdot\nabla\right){\bf u}
={\rm div}\left({\bm\sigma}\right)+\rho{\bf g}.
\end{equation}
We consider here both an incompressible flow and incompressible solid:
\begin{equation} \label{constitutive_model}
{\bm\sigma}={\bm\tau}-p{\bf I},
\end{equation}
with ${\bm\tau}$ being the deviatoric part of the stress tensor. For a Newtonian fluid in $\Omega^f_t$,
\begin{equation}
{\bm\tau}={\bm\tau}^f=\mu^f{\rm D}{\bf u}=\mu^f\left(\nabla {\bf u}+\nabla^{\scriptsize T} {\bf u}\right),
\end{equation}
and for a hyperelastic solid \cite{belytschko2013nonlinear} in $\Omega^s_t$,
\begin{equation} \label{constitutive_solid}
{\bm\tau}={\bm\tau}^s=J_{\mathcal{F}_t}^{-1}\frac{\partial\Psi\left({\bf F}\right)}{\partial{\bf F}}{\bf F}^T,
\end{equation}
with
\begin{equation}\label{definition_f}
{\bf F}
=\frac{\partial {\mathcal{F}\left(\hat{\bf x}, t\right)}}{\partial\hat{\bf x}}
\end{equation}
being the deformation tensor of the solid, $J_{\mathcal{F}_t}$ being the determinant of {\bf F}, and $\Psi\left({\bf F}\right)$ being the energy function of the hyperelastic solid material. Combining with the continuity equation
\begin{equation}\label{continuity}
\nabla\cdot{\bf u}=0 \quad {\rm in} \quad \Omega_t,
\end{equation}
the FSI system is completed with continuity of the velocity and normal stress conditions on the interface $\Gamma_t$:
\begin{equation}\label{interfaceBC1}
	{\bf u}^f={\bf u}^s, \quad {\bm \sigma}^f{\bf n}^f= {\bm \sigma}^s{\bf n}^f,
\end{equation}
and (for simplicity of this exposition) homogeneous Dirichlet and Neumann boundaries on $\Gamma_D$ and $\Gamma_N$ respectively:
\begin{equation}\label{dirichlet_neumann}
	{\bf u}=0, \quad {\bm \sigma}{\bf n}=0,
\end{equation}
with $\Gamma_D\cup\Gamma_N=\partial\Omega_t$ as shown in Figure \ref{aleformulation}.

\section{Finite element weak formulation}
\label{weak_formulation}
Let $L^2(\omega)$ be the square integrable functions in domain $\omega$, endowed with norm $\left\|u\right\|_{0,\omega}^2=\int_\omega \left|u\right|^2$. Let $H^1(\omega)=\left\{u:u \in L^2(\omega), \nabla u\in L^2(\omega)^d\right\}$ with the norm denoted by $\left\|u\right\|_{1,\omega}^2=\left\|u\right\|_{0,\omega}^2+\left\|\nabla u\right\|_{0,\omega}^2$. We also denote by $H_0^1\left(\omega\right)$ the
subspace of $H^1\left(\omega\right)$ whose functions have zero value on the Dirichlet boundary of $\omega$.

According to equation (\ref{ale_mapping}) we construct $\Omega_t$ from $\Omega_{t_0}$, so a function $v\in H_0^1(\Omega_t)$ is one-to-one corresponding to a function $\hat{v}\in H_0^1(\Omega_{t_0})$ via
\begin{equation}
v\circ\mathcal{A}_t=\hat{v}.
\end{equation}

Choosing a test function ${\bf v}\left( {\bf x}\right)={\bf v}\circ\mathcal{A}_t\left(\hat{\bf x}\right)=\hat{\bf v}\left(\hat{\bf x}\right)$, the weak formulation may be obtained by multiplying ${\bf v}$ on both sides of equation (\ref{momentum_equation_ale}), and integrating the stress term by parts in domain $\Omega_t^f$ and $\Omega_t^s$ separately:
\begin{equation}\label{weak_form1_f}
\begin{split}
&\rho^f\int_{\Omega_t^f}\frac{d{\bf u}\left(\mathcal{A}_t\left(\hat{\bf x}\right),t\right)}{dt}\cdot{\bf v}
+\rho^f\int_{\Omega_t^f}\left(\left({\bf u}-{\bf w}\right)\cdot\nabla\right){\bf u}\cdot{\bf v}\\
&+\frac{\mu^f}{2}\int_{\Omega_t^f}{\rm D}{\bf u}:{\rm D}{\bf v}
-\int_{\Omega_t^f}p\nabla\cdot {\bf v}
=\int_{\partial\Omega_t^f}{\bm{\sigma}}^f{\bf n}^f\cdot{\bf v}
+\rho^f\int_{\Omega_t^f}{\bf g}\cdot{\bf v}.
\end{split}
\end{equation}	

\begin{equation}\label{weak_form1_s}
\begin{split}
&\rho^s\int_{\Omega_t^f}\frac{d{\bf u}\left(\mathcal{A}_t\left(\hat{\bf x}\right),t\right)}{dt}\cdot{\bf v}
+\rho^s\int_{\Omega_t^f}\left(\left({\bf u}-{\bf w}\right)\cdot\nabla\right){\bf u}\cdot{\bf v} \\
&+\int_{\Omega_{t_0}^s}\frac{\partial\Psi}{\partial{\bf F}}:\nabla_{\hat{\bf x}}{\bf v}
-\int_{\Omega_t^s}p\nabla\cdot {\bf v}
=\int_{\partial\Omega_t^s}{\bm{\sigma}}^s\left(-{\bf n}^f\right)\cdot{\bf v}
+\rho^s\int_{\Omega_t^s}{\bf g}\cdot{\bf v}.
\end{split}
\end{equation}
We used $\frac{\partial\Psi}{\partial{\bf F}}{\bf F}^T:\nabla{\bf v}=\frac{\partial\Psi}{\partial{\bf F}}:\nabla{\bf v}{\bf F}=\frac{\partial\Psi}{\partial{\bf F}}:\nabla_{\hat{\bf x}}{\bf v}$ in the above deduction. Using the boundary conditions (\ref{interfaceBC1}) and (\ref{dirichlet_neumann}), we have the following equation by adding up (\ref{weak_form1_f}) and (\ref{weak_form1_s}).
\begin{equation}\label{weak_form1}
\begin{split}
&\rho\int_{\Omega_t}\frac{d{\bf u}\left(\mathcal{A}_t\left(\hat{\bf x}\right),t\right)}{dt}\cdot{\bf v}
+\rho\int_{\Omega_t}\left(\left({\bf u}-{\bf w}\right)\cdot\nabla\right){\bf u}\cdot{\bf v}\\
&+\frac{\mu^f}{2}\int_{\Omega_t^f}{\rm D}{\bf u}:{\rm D}{\bf v} 
-\int_{\Omega_t}p\nabla\cdot {\bf v}
+\int_{\Omega_{t_0}^s}\frac{\partial\Psi}{\partial{\bf F}}:\nabla_{\hat{\bf x}}{\bf v}
=\rho\int_{\Omega_t}{\bf g}\cdot{\bf v}.
\end{split}
\end{equation}	

Using Jacobi$^\prime$s formula \cite{magnus2019matrix}, we have
\begin{equation}
\begin{split}
\frac{\partial J_{\mathcal{A}_t}}{\partial t}
&=trace\left(J_{\mathcal{A}_t}{\bf A}^{-1}{\frac{\partial{\bf A}}{\partial t}}\right)\\
&=trace\left(J_{\mathcal{A}_t}{\bf A}^{-1}\nabla_{\hat{\bf x}}{\frac{\partial\mathcal{A}_t}{\partial t}}\right)\\
&=J_{\mathcal{A}_t}\nabla\cdot{\frac{\partial\mathcal{A}_t}{\partial t}}
=J_{\mathcal{A}_t}\nabla\cdot{\bf w},
\end{split}
\end{equation}
with 
${\bf A}
=\frac{\partial {\mathcal{A}\left(\hat{\bf x}, t\right)}}{\partial\hat{\bf x}}
=\nabla_{\hat{\bf x}}\mathcal{A}_t.
$
Then we can take the time derivative outside the moving domain (conservative formulation \cite{nobile1999stability}), 
\begin{equation}\label{dt_outside}
\begin{split}
&\frac{d}{dt}\int_{\Omega_t}{\bf u}\left({\bf x},t\right)\cdot{\bf v}\left({\bf x}\right)
=\frac{d}{dt}\int_{\Omega_{t_0}}J_{\mathcal{A}_t}{\bf u}\left(\mathcal{A}_t\left(\hat{\bf x}\right),t\right)\cdot\hat{\bf v}\left(\hat{\bf x}\right) \\
&=\int_{\Omega_t}\frac{d{\bf u}\left({\bf x},t\right)}{dt}\cdot{\bf v}\left({\bf x}\right)
+\int_{\Omega_t}\left(\nabla\cdot{\bf w}\right){\bf u}\left({\bf x},t\right)\cdot{\bf v}\left({\bf x}\right).
\end{split}
\end{equation}
Substituting (\ref{dt_outside}) into (\ref{weak_form1}), using 
\begin{equation}
{\rm div}\left({\bf w}\otimes{\bf u}\right)
=\left({\bf w}\cdot\nabla\right){\bf u}+\left(\nabla\cdot{\bf w}\right){\bf u},
\end{equation}
and combining the weak form of continuity equation (\ref{continuity}), leads to the weak formulation of the FSI problem:
\begin{problem}\label{problem_weak_continuous_c} 
Given $\Omega_{t_0}$, $\Gamma_{t_0}$, ${\bf u}(\hat{\bf x}, t_0)$ and an ALE mapping $\mathcal{A}_t$ (consequently given ${\bf w}$ by (\ref{definition_w})), $\forall\hat{\bf x}\in\Omega_{t_0}$: $\forall t\in(0,T]$ find ${\bf u}({\bf x},t)={\bf u}(\mathcal{A}_t\left(\hat{\bf x}\right),t)\in H_0^1(\Omega_t)^d$ and $p({\bf x},t)=p(\mathcal{A}_t\left(\hat{\bf x}\right),t)\in L^2(\Omega_t)$, such that $\forall {\bf v}({\bf x})={\bf v}(\mathcal{A}_t\left(\hat{\bf x}\right))$, ${\bf v}\in {H}_0^1(\Omega_t)^d$ and $\forall q({\bf x})=q(\mathcal{A}_t\left(\hat{\bf x}\right))$, $q\in L^2(\Omega_t)$, the following equations hold:
\begin{equation}\label{weak_form1_c}
\begin{split}
&\rho\frac{d}{dt}\int_{\Omega_t}{\bf u}\left(\mathcal{A}_t\left(\hat{\bf x}\right),t\right)\cdot{\bf v}
+\rho\int_{\Omega_t}\left({\bf u}\cdot\nabla\right){\bf u}\cdot{\bf v}
+\rho\int_{\Omega_t}\left({\bf w}\otimes{\bf u}\right):\nabla{\bf v} \\
&+\frac{\mu^f}{2}\int_{\Omega_t^f}{\rm D}{\bf u}:{\rm D}{\bf v} 
-\int_{\Omega_t}p\nabla\cdot {\bf v}
+\int_{\Omega_{t_0}^s}\frac{\partial\Psi}{\partial{\bf F}}\left({\bf F}\right):\nabla_{\hat{\bf x}}{\bf v}
=\rho\int_{\Omega_t}{\bf g}\cdot{\bf v},
\end{split}
\end{equation}	
\begin{equation}\label{weak_form2_c}
-\int_{\Omega_t} q\nabla \cdot {\bf u}=0,
\end{equation}
and
\begin{equation}
\mathcal{A}_t\left(\partial{\Omega_{t_0}}\right)=\mathcal{F}_t\left(\partial{\Omega_{t_0}}\right), \quad \mathcal{A}_t\left(\partial{\Gamma_{t_0}}\right)=\mathcal{F}_t\left(\partial{\Gamma_{t_0}}\right),
\end{equation}
with $\Gamma_{t_0}$ and $\partial\Omega_{t_0}$ being the initial interface and outer boundary respectively, as shown in Figure \ref{aleformulation}, and $\mathcal{F}_t$ being the Lagrangian mapping as defined in (\ref{lagrangian_mapping}).
\end{problem}

\section{Discretisation in space and time}
\label{section_discretization}
Define a stable finite element space, such as the Taylor-Hood elements, for the velocity-pressure pair $\left({\bf u}, p\right)$ in $\Omega_{t_0}$:
$$
V^h\left(\Omega_{t_0}\right)=span\left\{\hat{\varphi}_1,\cdots,\hat{\varphi}_{N^u}\right\} \subset H_0^1\left(\Omega_{t_0}\right)
$$
and
$$
L^h(\Omega_{t_0})=span\left\{\hat{\phi}_1,\cdots,\hat{\phi}_{N^p}\right\} \subset L^2\left(\Omega_{t_0}\right),
$$
with $N^u$ and $N^p$ being the number of nodal variables for each velocity component and pressure respectively. Then
$$
V^h\left(\Omega_t\right)=\left\{\varphi_h: \varphi_h=\hat{\varphi}_h\circ\mathcal{A}_t^{-1}, \hat{\varphi}_h\in V^h\left(\Omega_{t_0}\right)\right\},
$$
and
$$
L^h\left(\Omega_t\right)=\left\{\phi_h: \phi_h=\hat{\phi}_h\circ\mathcal{A}_t^{-1}, \hat{\phi}_h\in L^h\left(\Omega_{t_0}\right)\right\}.
$$

Using the backward Euler scheme, equation (\ref{weak_form1_c}) and (\ref{weak_form2_c}) can be discretised respectively as follows:
\begin{equation}\label{weak_form1_time1}
\begin{split}
&\frac{\rho}{\delta t}\int_{\Omega_{t_{n+1}}}{\bf u}_{n+1}^h\cdot{{\bf v}}
-\frac{\rho}{\delta t}\int_{\Omega_{t_n}}{\bf u}_n^h\cdot{{\bf v}} 
+\rho\int_{\Omega_{t_{n+1}}}\left({\bf u}_{n+1}^h\cdot\nabla\right){\bf u}^h_{n+1}\cdot{\bf v}\\
&+\rho {\mathcal{I}}\left(\xi(t)\right)
+\frac{\mu^f}{2}\int_{\Omega_{t_{n+1}}^f}{\rm D}{\bf u}_{n+1}^h:{\rm D}{\bf v}
-\int_{\Omega_{t_{n+1}}}p_{n+1}^h\nabla \cdot {\bf v} \\
&+\int_{\Omega_{t_0}^s}\frac{\partial\Psi}{\partial{\bf F}}\left({\bf F}_{n+1}\right):\nabla_{\hat{\bf x}}{{\bf v}}
=\int_{\Omega_{t_{n+1}}}\rho{\bf g}\cdot{\bf v},
\end{split}
\end{equation}
and
\begin{equation}\label{weak_form2_time1}
-\int_{\Omega_{t_{n+1}}} q\nabla\cdot {\bf u}_{n+1}^h=0.
\end{equation}
In the above
\begin{equation}
\xi(t)=\int_{\Omega_{t}}\left({\bf w}(t)\otimes{\bf u}_{n+1}^h\right):\nabla{\bf v},
\end{equation}
and $\delta t\mathcal{I}(\xi)$ is a quadrature formula used to compute $\int_{t_n}^{t_{n+1}}\xi(t)$. In order to have an unconditionally stable scheme, which will be proved in Section \ref{stability_energy}, the mid-point integration is adopted for
\begin{equation}\label{quadrature_2d}
{\mathcal{I}}\left(\xi\right)=\xi\left(t_{n+1/2}\right)
\end{equation}
in the two dimensional case, and the Simpson formula is adopted in the three dimensional case:
\begin{equation}\label{quadrature_3d}
{\mathcal{I}}\left(\xi\right)
=\frac{2}{3}\xi\left(t_{n+1/2}\right)
+\frac{1}{6}\xi\left(t_n\right)
+\frac{1}{6}\xi\left(t_{n+1}\right).
\end{equation}
Due to the definition of the deformation tensor ${\bf F}$ (\ref{definition_f}) and ALE velocity ${\bf w}$ (\ref{definition_w}), we have
\begin{equation}\label{expression_f_time}
\frac{{\bf F}_{n+1}-{\bf F}_n}{\delta t}
=\frac{{\bf F}_{n+1}\circ{\mathcal{F}}_{t_{n+1}}\left(\hat{\bf x}\right)
-{\bf F}_n\circ{\mathcal{F}}_{t_n}\left(\hat{\bf x}\right)}{\delta t}
\approx\nabla_{\hat{\bf x}}{\bf u}_{n+1},
\end{equation}
and
\begin{equation}\label{expression_w_time}
\frac{{\bf x}_{n+1}-{\bf x}_n}{\delta t}
=\frac{{\mathcal{A}}_{t_{n+1}}\left(\hat{\bf x}\right)
-{\mathcal{A}}_{t_n}\left(\hat{\bf x}\right)}{\delta t}
\approx{\bf w}_{n+1}.
\end{equation}
Therefore ${\bf F}_{n+1}$ and $\Omega_{t_{n+1}}$ in (\ref{weak_form1_time1}) can be updated as follows:
\begin{equation}\label{update_f}
{\bf F}_{n+1}={\bf F}_n+\delta t\nabla_{\hat{\bf x}}{\bf u}_{n+1},
\end{equation}
and
\begin{equation}\label{update_omega}
\Omega_{t_{n+1}}={\mathcal{A}}_{t_{n+1}}\left(\Omega_{t_0}\right)
=\left\{{\bf x}:{\bf x}={\bf x}_n+\delta t{\bf w}_{n+1},{\bf x}_n\in {\mathcal{A}}_{t_n}\left(\Omega_{t_0}\right)\right\}.
\end{equation} 

Up to now we have not stated how to construct ${\bf w}$ (or $\mathcal{A}_t$), because very often we only need to construct the ALE mapping $\mathcal{A}_t$ at a discrete time level, that is to say computing $\mathcal{A}_{t_{n+1}}$ for $n=0,1,\ldots$ at each time step. This will be explained in the rest of this section.

We solve the following static linear elastic equation in $\Omega_{t_{n+1}}$ in order to compute ${\bf w}_{n+1}$, and take ${\bf w}(t)={\bf w}_{n+1}$ for $t\in (t_n, t_{n+1}]$.
Given the following boundary data:

\begin{equation}\label{mesh_bc1}
{\bf w}_{n+1}\cdot{\bf n}=0 \qquad {\rm on} \quad \partial{\Omega_{t_{n+1}}},
\end{equation}
and
\begin{equation}\label{mesh_bc2}
{\bf w}_{n+1}={\bf u}_{n+1}^h \qquad {\rm on} \quad \Gamma_{t_{n+1}},
\end{equation}
find ${\bf w}_{n+1}\in V^h(\Omega_{t_{n+1}})^d$ such that $\forall {\bf z}\in V^h(\Omega_{t_{n+1}})^d$, the following equation holds:
\begin{equation} \label{mesh_equation}
\frac{\mu}{2}\int_{\Omega_{t_{n+1}}}{\rm D}{\bf w}_{n+1}:{\rm D}{\bf z}
+\lambda\int_{\Omega_{t_{n+1}}}\left(\nabla\cdot{\bf w}_{n+1}\right)\left(\nabla\cdot{\bf z}\right)=0,
\end{equation}
with $\mu$ and $\lambda$ being the Lam$\rm{\acute{e}}$ constants used here as pseudo-solid parameters. It is well known that the above elliptic problem (\ref{mesh_bc1}) to (\ref{mesh_equation}) has a unique solution ${\bf w}\in V^h\left(\Omega_{t_{n+1}}\right)$ \cite{brenner2007mathematical}. As a result, we are able to construct a mapping for $t\in(t_n, t_{n+1}]$,
\begin{equation}\label{construct_a_tn_t}
\mathcal{A}_{t_n,t}: \Omega_{t_n}\rightarrow\Omega_t, \quad \mathcal{A}_{t_n,t}\left({\bf x}_n\right)={\bf x}_n+(t-t_n){\bf w}_{n+1},
\end{equation} 
and further
\begin{equation}\label{ale_mapping_contruction}
\mathcal{A}_t=\mathcal{A}_{t_0,t_1}^{-1}\circ\mathcal{A}_{t_1,t_2}^{-1} \dots \circ \mathcal{A}_{t_n,t}^{-1}.
\end{equation}
From the computational point of view, knowing the ALE velocity ${\bf w}_{n+1}$ at the discrete level is sufficient.

Putting all the above together, the discrete ALE-FSI problem reads:
\begin{problem}\label{problem_weak_time} 
Given $\mathcal{A}_{t_n}$ and ${\bf u}_n^h={\bf u}(\mathcal{A}_{t_n}\left(\hat{\bf x}\right), t_n)$, $\forall\hat{\bf x}\in\Omega_{t_0}$ find ${\bf u}_{n+1}^h={\bf u}(\mathcal{A}_{t_{n+1}}\left(\hat{\bf x}\right),t_{n+1})\in V^h(\Omega_{t_{n+1}})^d$, $p_{n+1}^h=p(\mathcal{A}_{t_{n+1}}\left(\hat{\bf x}\right),t_{n+1})$ $\in L^h(\Omega_{t_{n+1}})$, and ${\bf w}_{n+1}\in V^h(\Omega_{t_{n+1}})^d$ (consequently an ALE mapping $\mathcal{A}_{t_{n+1}}$ by (\ref{ale_mapping_contruction})), such that $\forall {\bf v}({\bf x})={\bf v}(\mathcal{A}_{t_{n+1}}\left(\hat{\bf x}\right))$, ${\bf v}\in V^h(\Omega_{t_{n+1}})^d$, $\forall q({\bf x})=q(\mathcal{A}_{t_{n+1}}\left(\hat{\bf x}\right))$, $q\in L^h(\Omega_{t_{n+1}})$ and $\forall {\bf z}\in V^h(\Omega_{t_{n+1}})^d$, the following equation system holds:
\begin{equation}\label{weak_form_final}
\begin{split}
&\frac{\rho}{\delta t}\int_{\Omega_{t_{n+1}}}{\bf u}_{n+1}^h\cdot{{\bf v}}
-\frac{\rho}{\delta t}\int_{\Omega_{t_n}}{\bf u}_n^h\cdot{{\bf v}} 
+\rho\int_{\Omega_{t_{n+1}}}\left({\bf u}_{n+1}^h\cdot\nabla\right){\bf u}^h_{n+1}\cdot{\bf v}\\
&+\rho {\mathcal{I}}\left(\xi(t)\right)
+\frac{\mu^f}{2}\int_{\Omega_{t_{n+1}}^f}{\rm D}{\bf u}_{n+1}^h:{\rm D}{\bf v}
-\int_{\Omega_{t_{n+1}}}p_{n+1}^h\nabla \cdot {\bf v} \\
&-\int_{\Omega_{t_{n+1}}} q\nabla\cdot {\bf u}_{n+1}^h
+\int_{\Omega_{t_0}^s}\frac{\partial\Psi}{\partial{\bf F}}\left({\bf F}_{n+1}\right):\nabla_{\hat{\bf x}}{{\bf v}} \\
&+\frac{\mu}{2}\int_{\Omega_{t_{n+1}}}{\rm D}{\bf w}_{n+1}:{\rm D}{\bf z}
+\lambda\int_{\Omega_{t_{n+1}}}\left(\nabla\cdot{\bf w}_{n+1}\right)\left(\nabla\cdot{\bf z}\right)
=\int_{\Omega_{t_{n+1}}}\rho{\bf g}\cdot{\bf v}.
\end{split}
\end{equation}
with quadrature formula (\ref{quadrature_2d}) in 2D or (\ref{quadrature_3d}) in 3D, updating ${\bf F}_{n+1}$ by (\ref{update_f}) and updating $\Omega_{t_{n+1}}$ by (\ref{update_omega}). In addition, the above FSI system equations are completed with the Dirichlet and Neumann boundary conditions (\ref{dirichlet_neumann}) for the momentum and continuity equations (\ref{weak_form1_time1}) and (\ref{weak_form2_time1}), and with the boundary conditions (\ref{mesh_bc1}) and (\ref{mesh_bc2}) for the mesh equation (\ref{mesh_equation}).
\end{problem}

Problem \ref{problem_weak_time} is a highly non-linear system, so we solve it iteratively as described in the following Algorithm \ref{algorithm_problme2}.
\begin{algorithm}
\caption{Solve Problem \ref{problem_weak_time} for $\mathcal{A}_{t_{n+1}}$ (or ${\bf w}_{n+1}^h$), ${\bf u}_{n+1}^h$ and $p_{n+1}^h$}
\begin{algorithmic}\label{algorithm_problme2}
\REQUIRE $\Omega_{t_n}=\mathcal{A}_{t_n}\left(\Omega_{t_0}\right)$, ${\bf u}_n^h$ and a tolerance {\em tol}
\ENSURE  $\Omega_{t_{n+1}^k}=\Omega_{t_n}$, ${\bf u}_{n+1}^k={\bf u}_n^h$ and $k=0$
\REPEAT
\STATE 1. solve the mesh equation (\ref{mesh_equation}) for ${\bf w}_{n+1}^{k+1}$ using boundary conditions (\ref{mesh_bc1}) and (\ref{mesh_bc2})
\STATE 2. update $\Omega_{t_{n+1}^{k+1}}=\Omega_{t_{n+1}^{k}}+\delta t{\bf w}_{n+1}^{k+1}$ using (\ref{update_omega})
\STATE 3. solve the FSI system (\ref{weak_form1_time1}) and (\ref{weak_form2_time1}) for ${\bf u}_{n+1}^{k+1}$ and $p_{n+1}^{k+1}$
\STATE 4. $\epsilon_k = \frac{\|{\bf u}_{n+1}^{k+1}-{\bf u}_{n+1}^k\|}{\|{\bf u}_{n+1}^k\|}$,  $k \leftarrow k+1$
\UNTIL{$\epsilon_k < tol$}
\end{algorithmic}
\end{algorithm}
\section{Stability analysis}
\label{stability_energy}
We shall deduce an energy stability result at the end of this section. In preparation for this we first prove the following lemmas.

\begin{lemma}\label{lamma_div}
If $\left({\bf u}, p, {\bf w}\right)$ is the solution of Problem \ref{problem_weak_time}, then ${\bf u}$ satisfies the following at $t=t_{n+1}$.
\begin{equation}
\int_{\Omega_t}\left({\bf u}\cdot\nabla\right){\bf u}\cdot{\bf u}=0.
\end{equation}
\end{lemma}
\begin{proof}
Noticing that
\begin{equation}\label{deduce1}
\int_{\Omega_t}\left({\bf u}\cdot\nabla\right){\bf u}\cdot{\bf u}
=\int_{\Omega_t}\nabla\cdot\left({\bf u}\otimes{\bf u}\right)\cdot{\bf u}
-\int_{\Omega_t}\left|{\bf u}\right|^2\nabla\cdot{\bf u},
\end{equation}
and integrating by parts:
\begin{equation}\label{deduce2}
\begin{split}
& \int_{\Omega_t}\left({\bf u}\cdot\nabla\right){\bf u}\cdot{\bf u}
=\int_{\partial\Omega_t}\left|{\bf u}\right|^2{\bf u}\cdot{\bf n}
-\int_{\Omega_t}\left({\bf u}\cdot\nabla\right){\bf u}\cdot{\bf u}
-\int_{\Omega_t}\left|{\bf u}\right|^2\nabla\cdot{\bf u}.\\
&\Rightarrow
\int_{\Omega_t}\left({\bf u}\cdot\nabla\right){\bf u}\cdot{\bf u}
=\frac{1}{2}\int_{\partial\Omega_t}\left|{\bf u}\right|^2{\bf u}\cdot{\bf n}
-\frac{1}{2}\int_{\Omega_t}\left|{\bf u}\right|^2\nabla\cdot{\bf u}.
\end{split}
\end{equation}
In the above $\int_{\partial\Omega_t}\left|{\bf u}\right|^2{\bf u}\cdot{\bf n}=0$, thanks to the enclosed flow ${\bf u}\cdot{\bf n}=0$ (\ref{dirichlet_neumann}). Using the Sobolev imbedding theorem \cite[Theorem 6 in Chapter 5] {mitrovic1997fundamentals}, we have $H^1\subset L^{\infty}$ in the two dimensional case and $H^1\subset L^6$ in the three dimensional case. Either $L^\infty$ or $L^6$ is included in $L^4$ because $\Omega_t$ has finite measure. Therefore ${\bf u}\in H^1 \subset L^4 \Rightarrow \left|{\bf u}\right|^2\in L^2$, and $\int_{\Omega_t}\left|{\bf u}\right|^2\nabla\cdot{\bf u}=0$ thanks to (\ref{weak_form2_time1}).
\end{proof}

\begin{lemma}\label{lamma_f_lamma}
If $\left({\bf u}, p, {\bf w}\right)$ is the solution of Problem \ref{problem_weak_time} then, for any ${\bf w}\in V^h\left(\Omega_t\right)$, ${\bf u}$ satisfies the following at $t=t_{n+1}$.
\begin{equation}\label{lemma2_1}
\xi(t)\equiv\int_{\Omega_{t}}\left({\bf w}\otimes{\bf u}\right):\nabla{\bf u}
=-\frac{1}{2}\int_{\Omega_t}\left|{\bf u}\right|^2\nabla\cdot{\bf w}.
\end{equation}
\end{lemma}
\begin{proof}
Integrating by parts we get
\begin{equation}\label{lemma2_2}
\xi(t)
=\int_{\partial\Omega_{t}}\left({\bf w}\otimes{\bf u}\right){\bf u}\cdot{\bf n}
-\int_{\Omega_t}\nabla\cdot\left({\bf w}\otimes{\bf u}\right)\cdot{\bf u}
\end{equation}
The boundary integral in (\ref{lemma2_2}) is zero due to the enclosed flow ${\bf u}\cdot{\bf n}=0$ condition (\ref{dirichlet_neumann}). The second term on the right-hand side of (\ref{lemma2_2}) can be expressed as:
\begin{equation}\label{lemma2_3}
\int_{\Omega_t}\nabla\cdot\left({\bf w}\otimes{\bf u}\right)\cdot{\bf u}
=\xi(t)
+\int_{\Omega_t}\left|{\bf u}\right|^2\nabla\cdot{\bf w},
\end{equation}
we then have (\ref{lemma2_1}) by substituting (\ref{lemma2_3}) into (\ref{lemma2_2}).
\end{proof}

\begin{lemma}\label{lamma_quadrature_lamma}
If $\left({\bf u}_{n+1}, p_{n+1}, {\bf w}_{n+1}\right)$ is the solution of Problem \ref{problem_weak_time}, then
\begin{equation}\label{lemman_quadrature}
\|{\bf u}_{n+1}\|_{0,\Omega_{t_{n+1}}}^2
-\|{\bf u}_{n+1}\|_{0,\Omega_{t_n}}^2
=\delta t{\mathcal{I}}\left(\eta\right),
\end{equation}
with
\begin{equation}
\eta(t)=\int_{\Omega_t}\left|{\bf u}_{n+1}\right|^2\nabla\cdot{\bf w}(t), \quad  t\in\left(t_n, t_{n+1}\right).
\end{equation}
\end{lemma}
\begin{proof}
Since
\begin{eqnarray}
\begin{split}
&\eta(t)=\int_{\Omega_{t_n}}J_{\mathcal{A}_{t_n,t}}\left|{\bf u}_{n+1}\right|^2\left(\frac{\partial\mathcal{A}_{t_n,t}^{-1}}{\partial{\bf x}}\nabla_{{\bf x}_n}\right)\cdot{\bf w}(t) \\
& =\int_{\Omega_{t_n}}\left|{\bf u}_{n+1}\right|^2\left({\bf C}_{\mathcal{A}_{t_n, t}}\nabla_{{\bf x}_n}\right)\cdot{\bf w}(t),
\end{split}
\end{eqnarray}
where ${\bf C}_{\mathcal{A}_{t_n, t}}$ is the cofactor matrix of $\frac{\partial\mathcal{A}_{t_n, t}}{\partial{\bf x}}$. According to the way we construct ${\mathcal{A}_{t_n, t}}$ (\ref{construct_a_tn_t}), we know ${\bf C}_{\mathcal{A}_{t_n, t}}$ is a polynomial in time of degree $d-1$ \cite{nobile1999stability}, with $d=2,3$ being the space dimension. Also ${\bf w}(t)={\bf w}_{n+1}$ is a constant for $t\in (t_n,t_{n+1}]$, so $\eta(t)$ is linear in time when $d=2$ and quadratic when $d=3$, and a mid-point integration ($d=2$) or Simpson formula ($d=3$) would exactly compute $\int_{t_n}^{t_{n+1}}\eta(t)$. This is to say
\begin{equation}\label{lemma_qua_1}
{\mathcal{I}}\left(\eta\right)
=
\int_{t_n}^{t_{n+1}}\eta(t).
\end{equation}
Noticing that for $t\in\left(t_n, t_{n+1}\right)$,
\begin{equation}
\begin{split}
&\frac{d}{dt}\int_{\Omega_t}\left|{\bf u}_{n+1}\right|^2
=\frac{d}{dt}\int_{\Omega_{t_n}}J_{\mathcal{A}_{t_n,t}}\left|{\bf u}_{n+1}\right|^2 \\
&=\int_{\Omega_{t_n}}J_{\mathcal{A}_{t_n,t}}\left|{\bf u}_{n+1}\right|^2\nabla_{\bf x}{\bf w}(t)
=\eta(t),
\end{split}
\end{equation}
and using (\ref{lemma_qua_1}), we finally have (\ref{lemman_quadrature}).
\end{proof}

\begin{lemma}\label{lamma_phi}
Define potential energy of the solid:
\begin{equation}
E\left(t\right)
=\int_{\Omega_{t_0}^s}\Psi\left({\bf F}\right).
\end{equation}
If $\left({\bf u}_{n+1}, p_{n+1},{\bf w}_{n+1}\right)$ is the solution of Problem \ref{problem_weak_time} and $\Psi\left({\bf F}\right)$ is $C^1$ convex on the set of second order tensors \cite{Boffi_2016}, then
\begin{equation}\label{lemma_energy}
\delta t\int_{\Omega_{t_0}^s}\frac{\partial\Psi}{\partial{\bf F}}\left({\bf F}_{n+1}\right):\nabla_{\hat{\bf x}}{\bf u}_{n+1}
\ge
E\left(t_{n+1}\right)-E\left(t_n\right).
\end{equation}
\begin{proof}
Let 
\begin{equation}
w(t)=\Psi\left({\bf F}_n+t\left({\bf F}_{n+1}-{\bf F}_n\right)\right),
\end{equation}
then
\begin{equation}
w^\prime(t)=\frac{\partial\Psi}{\partial{\bf F}}\left({\bf F}_n+t\left({\bf F}_{n+1}-{\bf F}_n\right)\right):\left({\bf F}_{n+1}-{\bf F}_n\right).
\end{equation}
Due to the convexity assumption of $\Psi\left({\bf F}\right)$, we have
\begin{equation}
w^\prime(1)\ge w(1)-w(0).
\end{equation}
This gives:
\begin{equation}
\frac{\partial\Psi}{\partial{\bf F}}\left({\bf F}_{n+1}\right):\left({\bf F}_{n+1}-{\bf F}_n\right)
\ge
\Psi\left({\bf F}_{n+1}\right)
-\Psi\left({\bf F}_n\right).
\end{equation}
Using (\ref{update_f}) we have
\begin{equation}\label{energy_3}
\delta t\frac{\partial\Psi}{\partial{\bf F}}\left({\bf F}_{n+1}\right):\nabla_{\hat{\bf x}}{\bf u}_{n+1}
\ge
\Psi\left({\bf F}_{n+1}\right)
-\Psi\left({\bf F}_n\right).
\end{equation}
which finally leads to (\ref{lemma_energy}) by integrating (\ref{energy_3}) in $\Omega_{t_0}^s$.
\end{proof}
\end{lemma}

We now choose ${\bf v}={\bf u}_{n+1}^h$, $q=-p_{n+1}^h$ and ${\bf z}=0$ in equation (\ref{weak_form_final}) to deduce the stability result. Using Lemma \ref{lamma_div}, we have
\begin{equation}\label{energy_deduce_1}
\begin{split}
&\rho\int_{\Omega_{t_{n+1}}}{\bf u}_{n+1}^h\cdot{{\bf u}_{n+1}^h}
-\rho\int_{\Omega_{t_n}}{\bf u}_n^h\cdot{{\bf u}_{n+1}^h} \\
&+\delta t\rho {\mathcal{I}}\left(\xi(t)\right)
+\frac{\delta t\mu^f}{2}\int_{\Omega_{t_{n+1}}^f}{\rm D}{\bf u}_{n+1}^h:{\rm D}{\bf u}_{n+1}^h\\
&+\delta t\int_{\Omega_{t_0}^s}\frac{\partial\Psi}{\partial{\bf F}}\left({\bf F}_{n+1}\right):\nabla_{\hat{\bf x}}{\bf u}_{n+1}^h
=\delta t\int_{\Omega_{t_{n+1}}}\rho{\bf g}\cdot{\bf u}_{n+1}^h.
\end{split}
\end{equation}
Combining Lemmas \ref{lamma_f_lamma} and \ref{lamma_quadrature_lamma} we have
\begin{equation}\label{energy_deduce_2}
\|{\bf u}_{n+1}\|_{0,\Omega_{t_n}}^2
=\|{\bf u}_{n+1}\|_{0,\Omega_{t_{n+1}}}^2-\delta t{\mathcal{I}}\left(\eta\right)
=\|{\bf u}_{n+1}\|_{0,\Omega_{t_{n+1}}}^2+\delta t{\mathcal{I}}\left(\xi\right).
\end{equation}
Substituting equation (\ref{energy_deduce_2}) into the following estimate
\begin{equation}
\begin{split}
&\int_{\Omega_{t_{n}}}{\bf u}_{n}^h\cdot{\bf u}_{n+1}^h
\le
\|{\bf u}_n^h\|_{0,\Omega_{t_{n}}}\|{\bf u}_{n+1}^h\|_{0,\Omega_{t_{n}}}\\
&\le
\frac{1}{2}\left(\|{\bf u}_n^h\|_{0,\Omega_{t_{n}}}^2+\|{\bf u}_{n+1}^h\|_{0,\Omega_{t_{n}}}^2\right),
\end{split}
\end{equation}
we get
\begin{equation}\label{energy_deduce_3}
\int_{\Omega_{t_{n}}}{\bf u}_{n}^h\cdot{\bf u}_{n+1}^h
\le
\frac{1}{2}\left(\|{\bf u}_n^h\|_{0,\Omega_{t_{n}}}^2+\|{\bf u}_{n+1}^h\|_{0,\Omega_{t_{n+1}}}^2+\delta t\mathcal{I}(\xi)\right).
\end{equation}
Combining (\ref{energy_deduce_1}) and (\ref{energy_deduce_3}), and thanks to Lemma \ref{lamma_phi} the energy stability result reads: 
\begin{proposition} [Energy non-increasing] \label{lec_backward_Euler_space}
Let $\left({\bf u}_{n+1}^h, p_{n+1}^h,{\bf w}_{n+1}^h\right)$ be the solution of Problem \ref{problem_weak_time}, if there is no body force, then
\begin{equation}\label{energy_estimate_after_time_discretization_backward_Euler_space}
\begin{split}
&\frac{\rho}{2}\|{\bf u}_{n+1}^h\|_{0,\Omega_{t_{n+1}}}^2
+E\left(t_{n+1}\right)  
+\frac{\delta t\mu^f}{2}\sum_{k=1}^{n+1}\int_{\Omega_{t_k}^f}{\rm D}{\bf u}_k^h:{\rm D}{\bf u}_k^hd{\bf x}\\
&\le \frac{\rho}{2}\|{\bf u}_{n}^h\|_{0,\Omega_{t_{n}}}^2
+E\left(t_{n}\right)  
+\frac{\delta t\mu^f}{2}\sum_{k=1}^{n}\int_{\Omega_{t_k}^f}{\rm D}{\bf u}_k^h:{\rm D}{\bf u}_k^hd{\bf x}.
\end{split}
\end{equation}
\end{proposition}
The above estimate indicate that the total energy, including kinetic energy, potential energy and the viscous dissipation, of the FSI system is non-increasing.

\section{Implementation: {\bf F}-scheme and {\bf d}-scheme}
\label{implementation}
In this section, we focus on the implementation of a specific solid model, which determines the following term
\begin{equation}\label{solid_term}
\int_{\Omega_{t_0}^s}\frac{\partial\Psi}{\partial{\bf F}}\left({\bf F}_{n+1}\right):\nabla_{\hat{\bf x}}{{\bf v}}
\end{equation}
in equation (\ref{weak_form_final}). We consider an incompressible neo-Hookean solid model with the energy function $\Psi$ being given as follows \cite{Hesch_2014}:
\begin{equation}\label{energy_funciton_neo_hookean}
\Psi\left({\bf F}\right)=\frac{c_1}{2}\left[tr\left({\bf F}{\bf F}^T\right)-d-2ln\left(J_{\mathcal{F}_t}\right)\right].
\end{equation}
In order to compute the derivative of $\Psi$ with respective to ${\bf F}$, we first have
\begin{equation}\label{ff_derivative}
\begin{split}
&\left[\frac{\partial tr\left({\bf F}^T{\bf F}\right)}{\partial {\bf F}}\right]_{mn} 
=\frac{\partial tr\left(F_{ki}F_{kj}\right)}{\partial F_{mn}}
=\frac{\partial \sum_{k}^{d}\sum_{i}^{d}F_{ki}^2}{\partial F_{mn}} \\
&=\frac{\partial\left(F_{11}^2+F_{12}^2+\cdots+F_{dd}^2\right)}{\partial F_{mn}}
=2{\bf F}_{mn}.
\end{split}
\end{equation}
Let $cof(F_{ij})=(-1)^{i+1}det\left({\bf F}\right.$ without $i^{th}$ row and $j^{th}$ column$\left.\right)$ be the cofactor of $F_{ij}$. Because of $J_{\mathcal{F}_t}=\sum_{k}^{d}F_{ik}cof\left(F_{ik}\right)$, we have 
$
\frac{\partial J_{\mathcal{F}_t}}{\partial F_{ij}}
=cof\left(F_{ij}\right)
$, i.e,
\begin{equation}\label{jacobi_special_case}
\frac{\partial J_{\mathcal{F}_t}}{\partial {\bf F}}
=cof\left({\bf F}\right)
=J_{\mathcal{F}_t}{\bf F}^{-T}.
\end{equation}
Combining equations (\ref{ff_derivative}) and (\ref{jacobi_special_case}) gives
\begin{equation}
\frac{\partial{\Psi}}{\partial{\bf F}}
=c_1\left({\bf F}-{\bf F}^{-T}\right).
\end{equation}
Using formula (\ref{update_f}), the term (\ref{solid_term}) can then be expressed as:
\begin{equation}\label{f_formulation}
\begin{split}
&\int_{\Omega_{t_0}^s}\frac{\partial{\Psi}}{\partial{\bf F}}\left({\bf F}_{n+1}\right):\nabla_{\hat{\bf x}}{{\bf v}}
=c_1\int_{\Omega_{t_0}^s}\left({\bf F}_{n+1}-{\bf F}_{n+1}^{-T}\right):\nabla_{\hat{\bf x}}{{\bf v}}\\
&=c_1\int_{\Omega_{t_0}^s}{\bf F}_{n+1}:\nabla_{\hat{\bf x}}{{\bf v}}
-c_1\int_{\Omega_{t_{n+1}}^s}J_{\mathcal{F}_t}^{-1}\nabla\cdot{{\bf v}}\\
&=c_1\delta t\int_{\Omega_{t_0}^s}\nabla_{\hat{\bf x}}{\bf u}_{n+1}:\nabla_{\hat{\bf x}}{\bf v}
+c_1\int_{\Omega_{t_0}^s}{\bf F}_n:\nabla_{\hat{\bf x}}{{\bf v}}
-c_1\int_{\Omega_{t_{n+1}}^s}J_{\mathcal{F}_t}^{-1}\nabla\cdot{{\bf v}}.
\end{split}
\end{equation}

In the above we update the solid deformation tensor ${\bf F}$ and integrate in the initial configuration, and we call this the ${\bf F}$-scheme. We can also express the stress in terms of displacement ${\bf d}$ and integrate in the current configuration as introduced in \cite{Hecht_2017}, which is called the ${\bf d}$-scheme. To deduce the ${\bf d}$-scheme, we first transform the term (\ref{solid_term}) to be integrated in the current domain:
\begin{equation}
\int_{\Omega_{t_0}^s}\frac{\partial\Psi}{\partial{\bf F}}\left({\bf F}_{n+1}\right):\nabla_{\hat{\bf x}}{{\bf v}}
=\int_{\Omega_{t_{n+1}}^s}J_{\mathcal{F}_t}^{-1}\frac{\partial\Psi}{\partial{\bf F}}{\bf F}^T:\nabla{\bf v}
=\int_{\Omega_{t_{n+1}}^s}{\bm{\tau}^s}:\nabla{\bf v},
\end{equation}
where
\begin{equation}\label{neo_hookean_tau}
{\bm\tau^s}
=c_1J_{\mathcal{F}_t}^{-1}\left({\bf B}-{\bf I}\right)
\end{equation}
is the deviatoric stress tensor, with ${\bf B}={\bf F}{\bf F}^T$.

Let us only consider a two dimensional case, readers may refer to \cite{chiang2017numerical} for the three dimensional case. According to the Cayley-Hamilton theorem, ${\bf B}$ satisfies its characteristic equation:
\begin{equation}
{\bf B}^2-tr_{\bf B}{\bf B}+J_{\mathcal{F}_t}^2{\bf I}=0,
\end{equation}
from which we immediately have:
\begin{equation}\label{expression_b}
{\bf B}=tr_{\bf B}{\bf I}-J_{\mathcal{F}_t}^2{\bf B}^{-1}.
\end{equation}
Since
\begin{equation}
{\bf F}={\nabla_{\hat{\bf x}}{\bf x}}={\nabla_{\hat{\bf x}}({\hat{\bf x}}+{\bf d})}={\bf I}+{\bf F}\nabla{\bf d},
\end{equation}
we also have:
\begin{equation}\label{f_inverse}
{\bf F}^{-1}={\bf I}-\nabla{\bf d}.
\end{equation}
Substituting (\ref{expression_b}) and (\ref{f_inverse}) into (\ref{neo_hookean_tau}), ${\bm\tau}^s$ can be expressed by displacement as follows:
\begin{equation}
{\bm\tau}^s=-c_1J_{\mathcal{F}_t}\left({\bf I}-\nabla{\bf d}\right)^T\left({\bf I}-\nabla{\bf d}\right)
+c_1J_{\mathcal{F}_t}^{-1}\left(tr_{\bf B}-1\right){\bf I},
\end{equation}
which can further be written as
\begin{equation}\label{solid_stress_d}
{\bm\tau}^s=c_1J_{\mathcal{F}_t}\left({\rm D}{\bf d}-\nabla^{\rm T}{\bf d}\nabla{\bf d}\right)+\bar{p}{\bf I},
\end{equation}
where $\bar{p}=c_1J_{\mathcal{F}_t}^{-1}\left(tr_{\bf B}-1\right)-c_1J_{\mathcal{F}_t}$ will be integrated into the solid pressure $p$ in (\ref{constitutive_model}) as an unknown. Similarly to the update of ${\bf F}$ in (\ref{update_f}), updating the displacement by
\begin{equation}
{\bf d}_{n+1}=\tilde{\bf d}_n+\delta t{\bf u}_{n+1}, \quad \tilde{\bf d}_n={\bf d}_n\circ\mathcal{A}_{t_n,t_{n+1}}^{-1},
\end{equation}
leads to the computation of term (\ref{solid_term}) as follows:
\begin{equation}
\begin{split}
&\int_{\Omega_{t_0}^s}\frac{\partial\Psi}{\partial{\bf F}}\left({\bf F}_{n+1}\right):\nabla_{\hat{\bf x}}{{\bf v}}
=\int_{\Omega_{t_{n+1}}^s}{\bm{\tau}^s}:\nabla{\bf v}\\
&=c_1\int_{\Omega_{t_{n+1}}^s}\left({\rm D}{\bf d}_{n+1}-\nabla^{\rm T}{\bf d}_{n+1}\nabla{\bf d}_{n+1}\right):\nabla{\bf v}\\
&=\frac{c_1\delta t}{2}\int_{\Omega_{t_{n+1}}^s}{\rm D}{\bf u}_{n+1}:{\rm D}{\bf v}
+\frac{c_1}{2}\int_{\Omega_{t_{n+1}}^s}{\rm D}\tilde{\bf d}_n:{\rm D}{\bf v}  \\
&-\delta tc_1\int_{\Omega_{t_{n+1}}^s}\left(\nabla^{\rm T}{\bf u}_{n+1}\nabla\tilde{\bf d}_n+\nabla^{\rm T}\tilde{\bf d}_n\nabla{\bf u}_{n+1}\right):\nabla{\bf v}\\
&-c_1\int_{\Omega_{t_{n+1}}^s}\nabla^{\rm T}\tilde{\bf d}_n\nabla\tilde{\bf d}_n:\nabla{\bf v}.
\end{split}
\end{equation}
Note that in the above, the second order term $O\left(\delta t^2\right)$ is neglected and $J_{\mathcal{F}_t}$ is replaced by $1$.
This is justified through observations from numerical simulation \cite{Hecht_2017}.

\begin{remark}
The two and three dimensional ${\bf F}$-scheme have exactly the same formulations. This can been seen from equation (\ref{f_formulation}), which does not depend on dimensions. However the formulation of ${\bf d}$-scheme depends on the Cayley-Hamilton theorem, which is different in two and three dimensions, and consequently leads to significant complexity of the ${\bf d}$-scheme in three dimension \cite{chiang2017numerical}.
\end{remark}

\section{Numerical experiments}
\label{sec_numerical_exs}
In this section, we validate the proposed numerical scheme through a selection of benchmarks in the FSI area. We shall use the Taylor-Hood elements for the velocity-pressure pair. We validate the energy stability expressed by (\ref{energy_estimate_after_time_discretization_backward_Euler_space}) in Section \ref{sec_oscillating}. We validate the proposed scheme against a FSI problem with a semi-analytic solution in Section \ref{sec_rotating_disc}. Time and mesh convergence tests are carried out in Section \ref{sec_flag}, and an example with very large solid deformation is tested in Section \ref{sec_fallingdisc}. The {\bf F}-scheme will be adopted in all the following numerical tests. In addition, the {\bf d}-scheme is also implemented for tests in Section \ref{sec_oscillating} and \ref{sec_fallingdisc} in order to compare the two schemes.

\subsection{Oscillating disc}
\label{sec_oscillating}
In this test, we consider an enclosed flow (${\bf n}\cdot{\bf u}=0$) in $\Omega=[0,1]\times[0,1]$ with a periodic boundary condition. A solid disc is initially located in the middle of the square $\Omega$ and has a radius of $0.2$. The initial velocity of the fluid and solid are prescribed by the following stream function
\begin{equation*}
\Psi=\Psi_0{\rm sin}(ax){\rm sin}(by),
\end{equation*}
where $\Psi_0=5.0\times10^{-2}$ and $a=b=2\pi$. In this test, $\rho^f=1$, $\mu^f=0.01$, $\rho^s=1.5$ and $c_1=1$. A mesh size of 3217 elements with 13081 nodes is used in this test. In order to visualize the flow a snapshot ($t=0.25$) of the velocity and pressure field are presented in Figure \ref{snapshot_of_fluid}, and the evolution of energy is presented in Figure \ref{energy_evolution} and \ref{total_energy} from which we can observe the property of non-increasing total energy as proved in Proposition \ref{lec_backward_Euler_space}. 

The {\bf F}-scheme and {\bf d}-scheme are compared using this example and we have not found any significant difference by comparing the solid deformation as shown in Figure \ref{compare_d_f}.

\begin{figure}[h!]
	\begin{minipage}[t]{0.5\linewidth}
		\centering  
		\includegraphics[width=2.2in,angle=0]{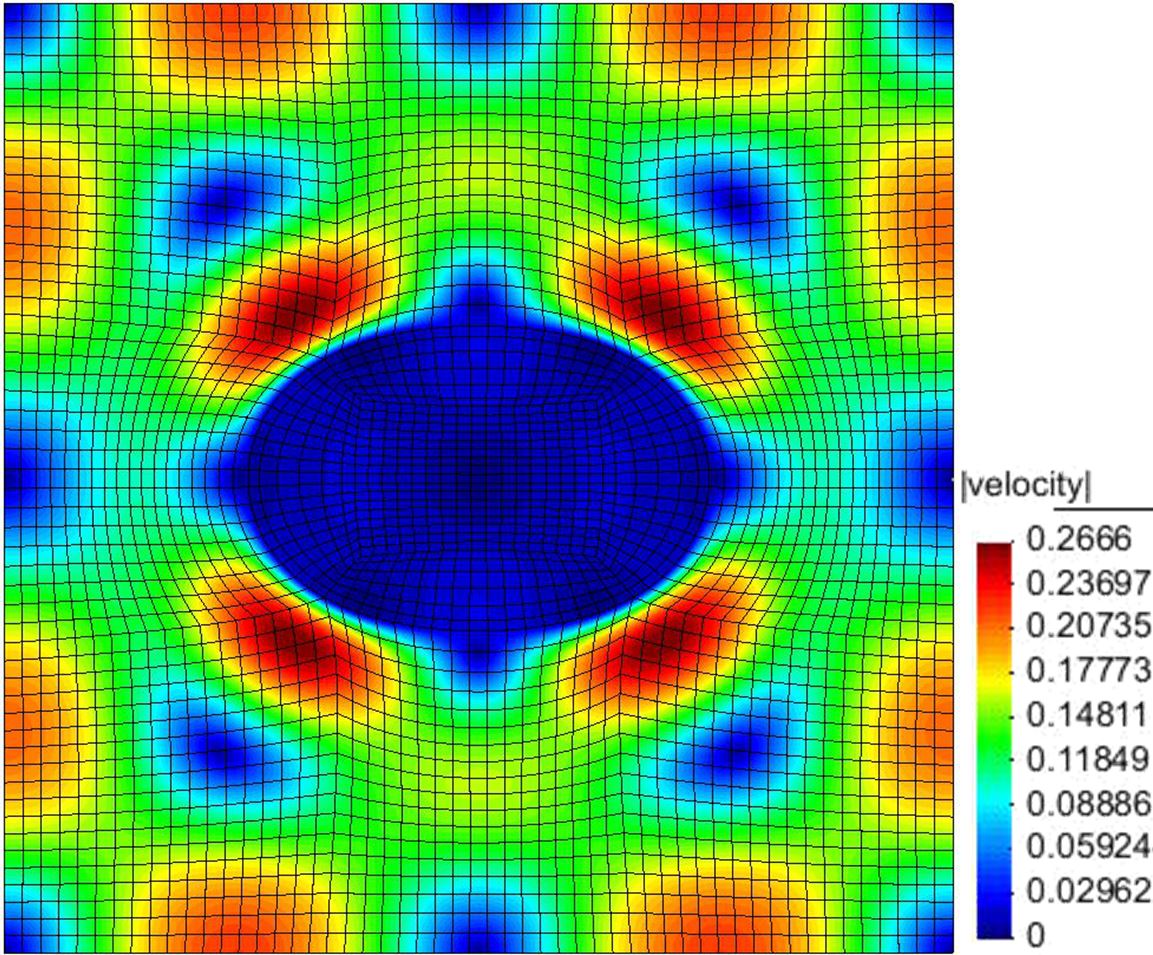}	
		\captionsetup{justification=centering}
		\caption*{\scriptsize(a) Velocity norm.}
	\end{minipage}
	\begin{minipage}[t]{0.5\linewidth}
		\centering  
		\includegraphics[width=2.3in,angle=0]{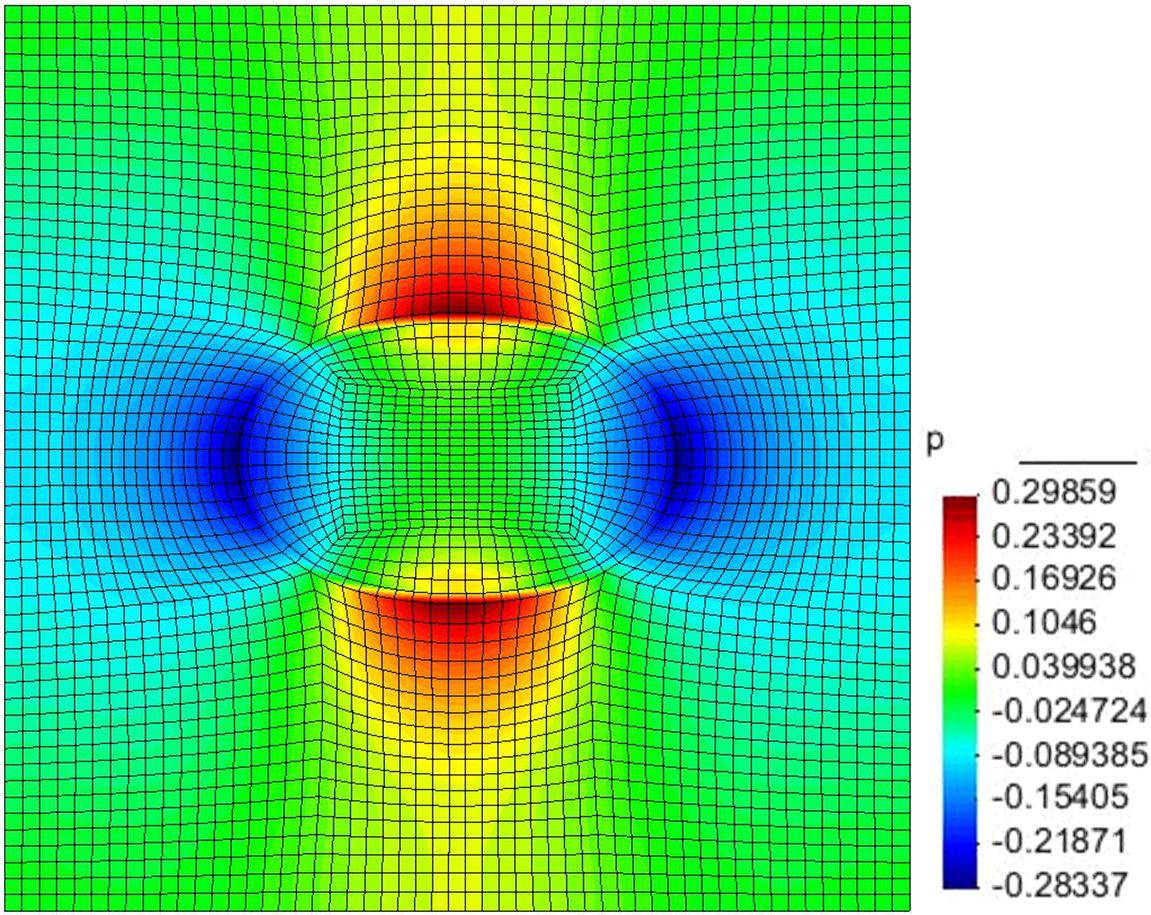}	
		\captionsetup{justification=centering}
		\caption*{\scriptsize(b) Pressure.}
	\end{minipage}     		
	\captionsetup{justification=centering}
	\caption {\scriptsize Snapshot of the oscillating disc at $t=0.25$ when the disc is maximally stretched, using a time step of $\Delta t=0.01$.} 
	\label{snapshot_of_fluid}
\end{figure}

\begin{figure}[h!]
	\centering  
	\includegraphics[width=2.7in,angle=0]{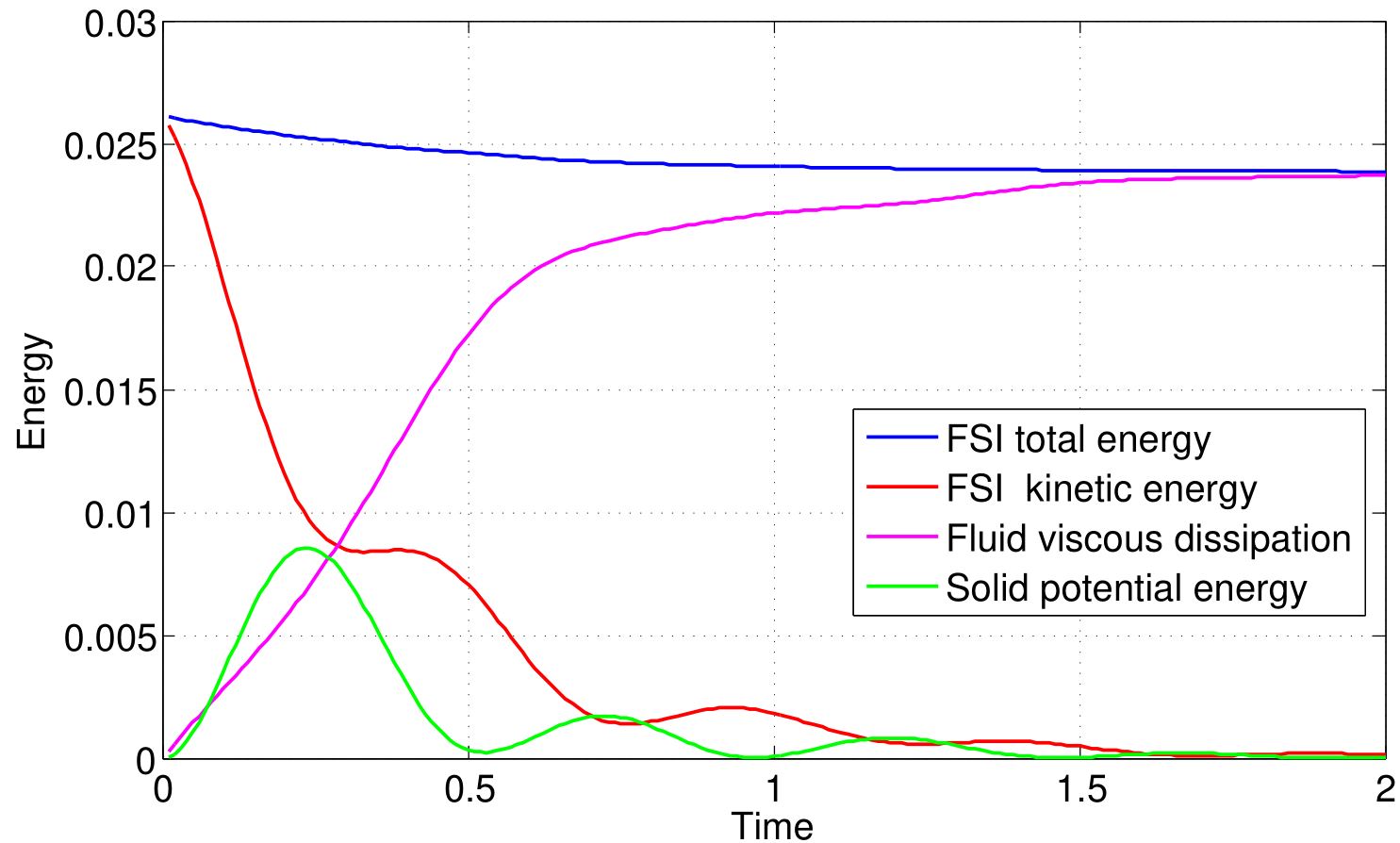}	    		
	\captionsetup{justification=centering}
	\caption {\scriptsize Evolution of energy for the oscillating disc using $\Delta t=0.01$. The peaks of the green curve indicate the time when the disc is maximally stretched. The first peak is horizontally stretched and the second peak is vertically stretched. The troughs of the green curve are the stress-free stages.} 
	\label{energy_evolution}
\end{figure}

\begin{figure}[h!]
\centering  
\includegraphics[width=2.7in,angle=0]{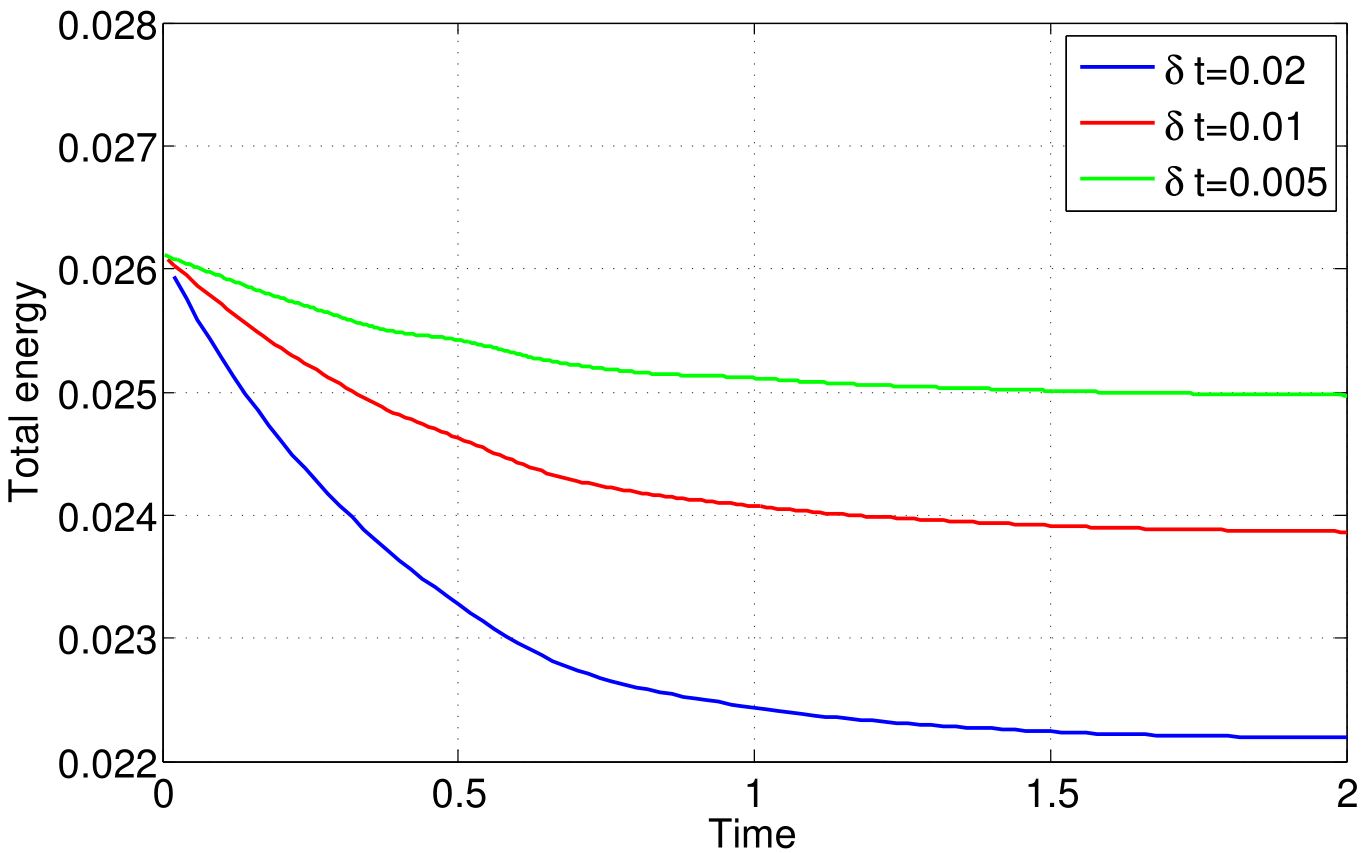}	    		
\captionsetup{justification=centering}
\caption {\scriptsize Evolution of total energy for the oscillating disc.} 
\label{total_energy}
\end{figure}

\begin{figure}[h!]
	\centering  
	\includegraphics[width=2.3in,angle=0]{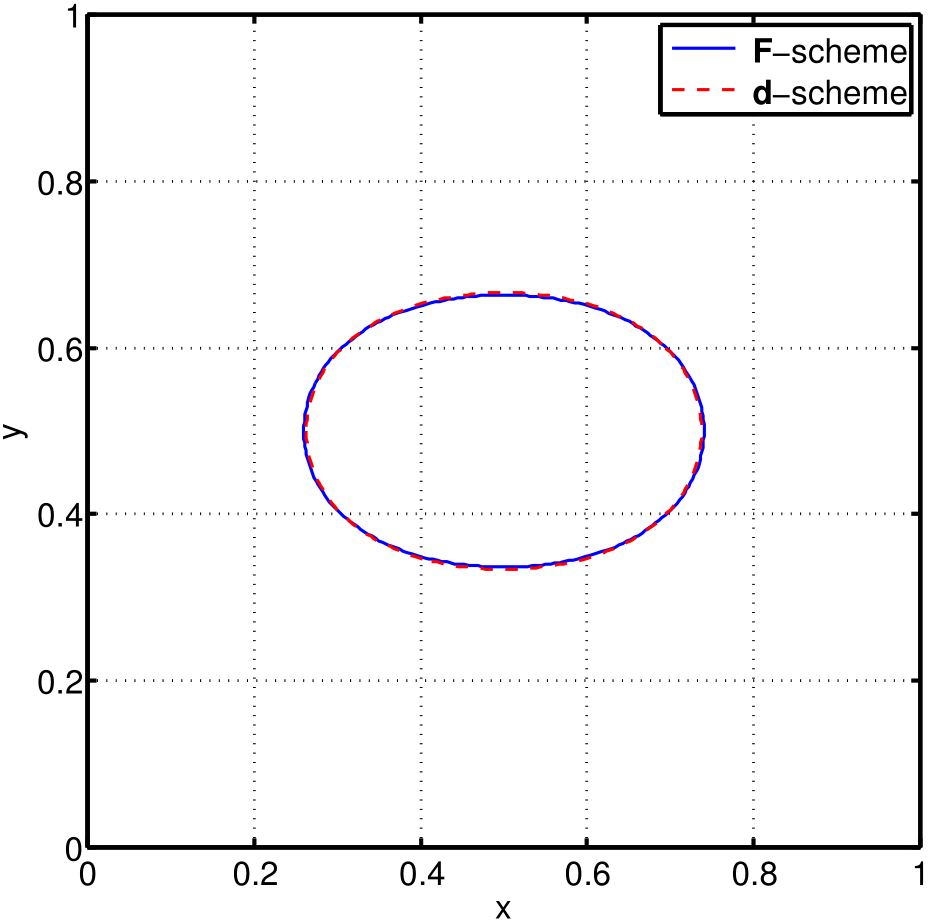}	    		
	\captionsetup{justification=centering}
	\caption {\scriptsize Comparison of disc shape for the {\bf F}-scheme and {\bf d}-scheme at $t=0.25$ when the disc is maximally stretched.} 
	\label{compare_d_f}
\end{figure}

\subsection{Rotating disc}
\label{sec_rotating_disc}
This test is taken from \cite{Hecht_2017}. The computational domain is the area between two concentric circles ($R_0$ and $R_1$) as shown in Figure \ref{rotaging_disc}, with fluid and solid properties as $\rho^f=1,\rho^s=2,\mu^f=2$ and $c_1=4$ . A constant angular velocity ($\omega=U/R_1=0.6$) is prescribed at the outer boundary. This velocity first induces the fluid, that is initially at rest, to rotate and then gradually drags the solid to rotate as well. Using the property of symmetry, this problem can be reduced to a one-dimensional equation when considered in a polar coordinate system ($r$, $\theta$) \cite{Hecht_2017}:
\begin{equation}\label{readuced_oned_rotating1}
\rho^f\frac{\partial u_\theta}{\partial t}
=\frac{\mu^f}{r}\frac{\partial}{\partial r}\left(r\frac{\partial u_\theta}{\partial r}\right)
-\mu^f\frac{u_{\theta}}{r^2}, \quad  \quad\quad R\le r<R_1
\end{equation}
and
\begin{equation}\label{readuced_oned_rotating2}
	\rho^s\frac{\partial u_\theta}{\partial t}
	=\frac{c_1}{r}\frac{\partial}{\partial r}\left(r\frac{\partial d_\theta}{\partial r}\right)
	-c_1\frac{d_{\theta}}{r^2}, \quad \frac{\partial d_\theta}{\partial t}=u_\theta, \quad\quad R_0< r \le R,
\end{equation}
where $u_r$ and $u_\theta$ are the velocity components in the radial and tangential directions respectively. This one-dimensional problem (\ref{readuced_oned_rotating1}) and (\ref{readuced_oned_rotating2}) can be solved to high accuracy, and the solution is plotted in Figure \ref{rotating_disc_oneD_results} using 200 linear elements and $\Delta t=1.0\times 10^{-3}$. Using the same time step, which is stable, the proposed method can produce results of similar accuracy to the semi-analytic solution (see Figure \ref{rotating_disc_solution}). We use three different meshes to test convergence of the proposed algorithm. A coarse mesh equally divides the radial direction of the computational domain into 4 segments, and equally divides the tangential direction into to 40 segments, which therefore has $4\times40=160$ biquadratic elements. The medium and fine mesh are refined based on the coarse mesh, which have $8\times80=640$ and $16\times160=2560$ elements respectively. Due to the discontinuity in the derivative at the fluid-solid interface, we only achieve an $O(h)$ convergence as shown in Figure \ref{rotating_disc_err}, where $h$ is the mesh size. This observation is consistent with the result in \cite{Hecht_2017}.

\begin{figure}[h!]
	\centering
	\includegraphics[width=3 in,angle=0]{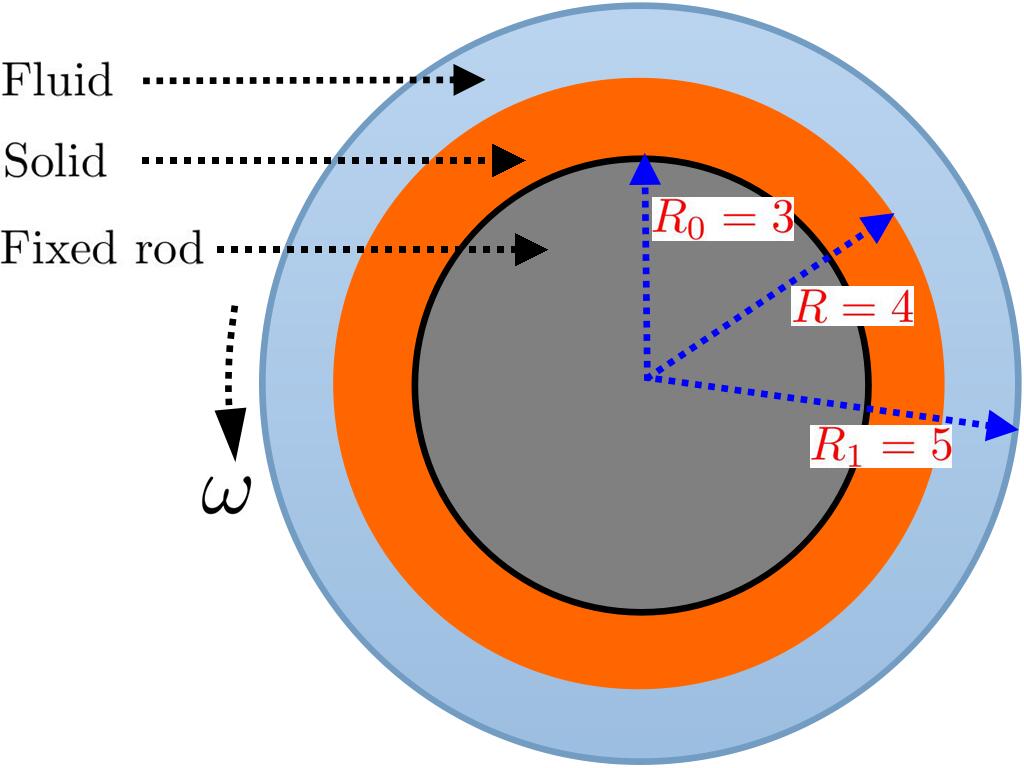}
	\captionsetup{justification=centering}
	\caption {\scriptsize Sketch of a rotating disc in Section \ref{sec_rotating_disc}.} 
	\label{rotaging_disc}
\end{figure}

\begin{figure}[h!]
	\centering  
	\includegraphics[width=2.8in,angle=0]{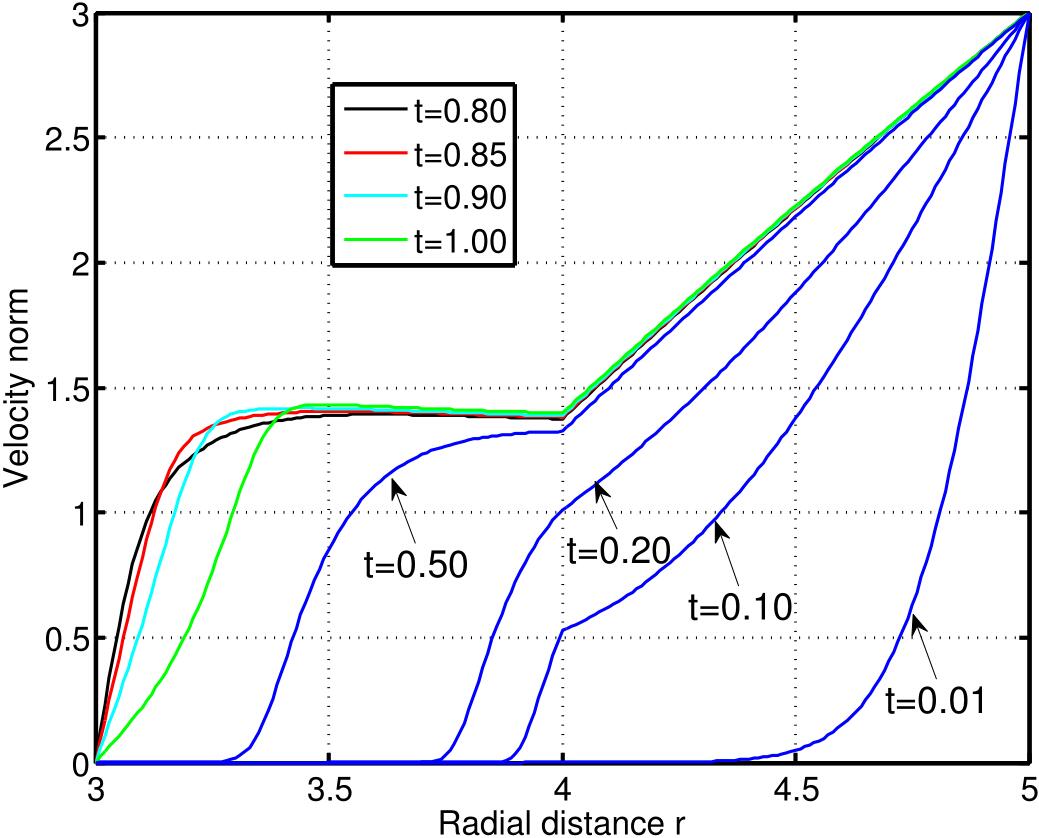}
	\captionsetup{justification=centering}	
	\caption {\scriptsize Evolution of the velocity norm for the reduced one-dimensional rotating disc.} 
	\label{rotating_disc_oneD_results}
\end{figure}

\begin{figure}[h!]
	\centering
	\includegraphics[width=2.8in,angle=0]{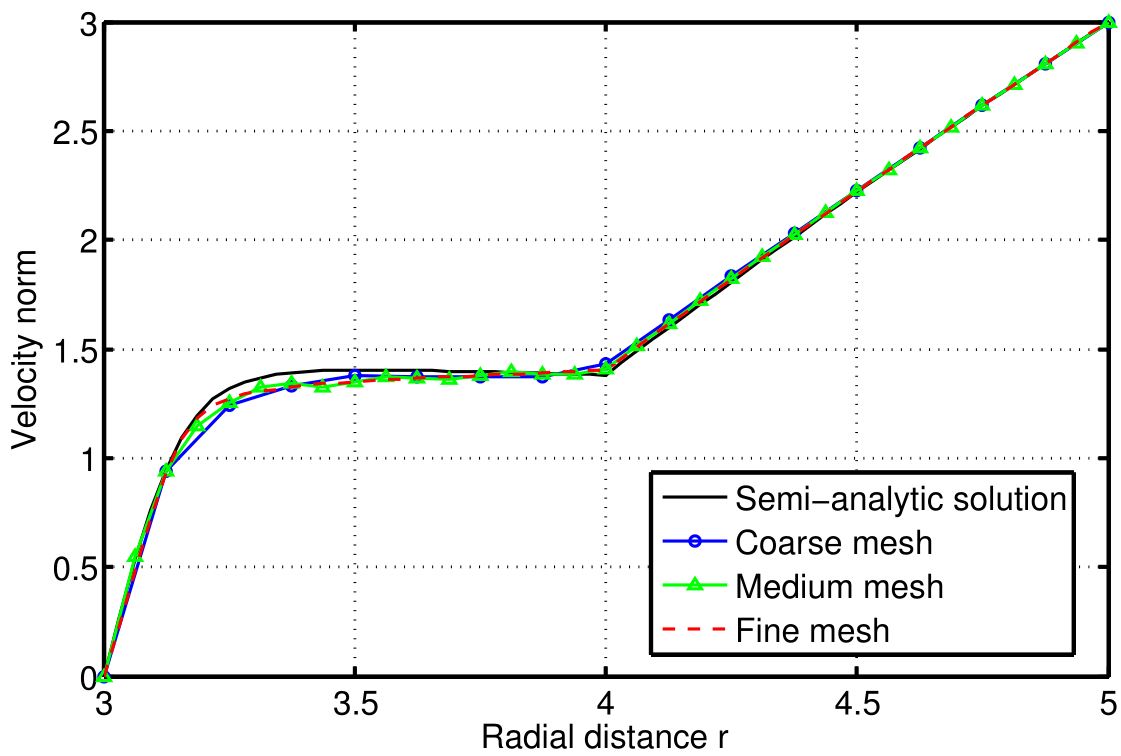}
	\captionsetup{justification=centering}	
	\caption {\scriptsize Comparison between the proposed approach and the semi-analytic solution at $t=0.85$ when the solid is maximally deformed.} 
	\label{rotating_disc_solution}	
\end{figure}

\begin{figure}[h!]
	\centering
	\includegraphics[width=2.8in,angle=0]{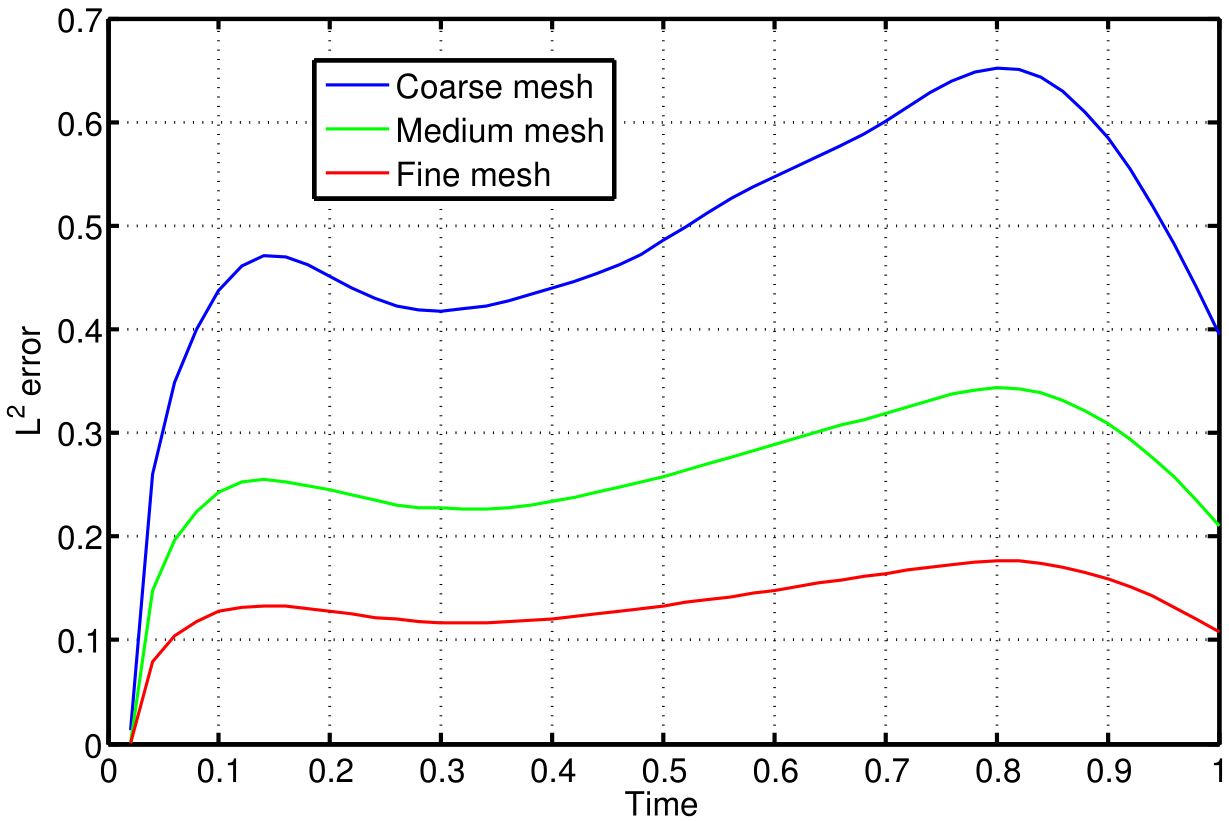}
	\captionsetup{justification=centering}	
	\caption {\scriptsize Convergence of $L^2$ error.} 
	\label{rotating_disc_err}	
\end{figure}

\subsection{Oscillating flag}
\label{sec_flag}
In this section, we consider an oscillating flag attached to a cylinder, which was firstly proposed in \cite{turek2006proposal} (name FSI3), and been regarded as a challenging numerical test in the FSI field. We test the time and mesh convergence for the proposed FSI method. The computational domain is a rectangle ($L\times H$) with a cut hole of radius $r$ and center $(c,c)$ as shown in Figure \ref{thick_flag_sketch}. A leaflet of size $l\times h$ is attached to the boundary of the hole (the mesh of the leaflet is fitted to the boundary of the hole, see the solid mesh in Figure \ref{flag}). In this test, $L=2.5$, $H=0.41$, $l=0.35$, $h=0.02$, $c=0.2$ and $r=0.05$. The fluid and solid parameters are as follows: $\rho^f=\rho^s=10^3$, $\mu^f=1$ and $c_1=2.0\times 10^6$. The inlet flow is prescribed as:
\begin{equation}
	\bar{u}_x=\frac{12y}{H^2}\left(H-y\right),\quad \bar{u}_y=0.
\end{equation}

\begin{figure}[h!]
	\centering  
	\includegraphics[width=4.0in,angle=0]{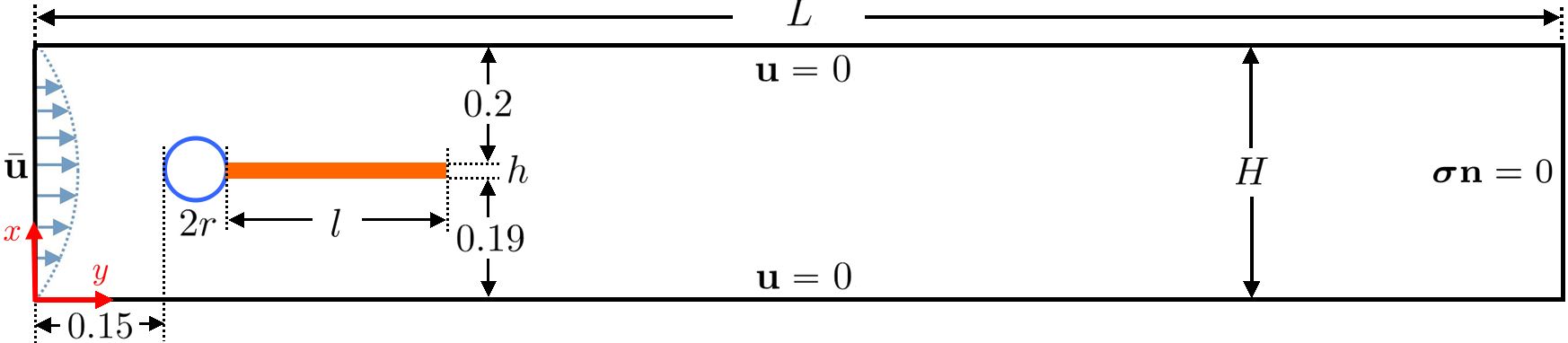}			
	\captionsetup{justification=centering}
	\caption {\scriptsize Computational domain and boundary conditions for the oscillating flag.} 
	\label{thick_flag_sketch}
\end{figure}

A wall boundary condition and the outlet flow condition are displayed in Figure \ref{thick_flag_sketch}. A coarse mesh has $10054$ nodes and $2448$ biquadratic elements as shown in Figure \ref{flag}, and a medium and fine mesh have 33746 nodes (8320 elements) and 68974 nodes (17081 elements) respectively. We study the oscillating frequency and amplitude at the tip of the flag. The convergence with respect to time and space are displayed in Figure \ref{flag_time} and Figure \ref{flag_mesh} respectively, and the frequency and amplitude of the oscillation converge to 5.26 and 0.035 respectively. These figures have a good agreement with the reference values given in \cite{turek2006proposal} with frequency and amplitude being 5.3 and 0.03438 respectively.

\begin{figure}[h!]
	\centering  
	\includegraphics[width=4.5in,angle=0]{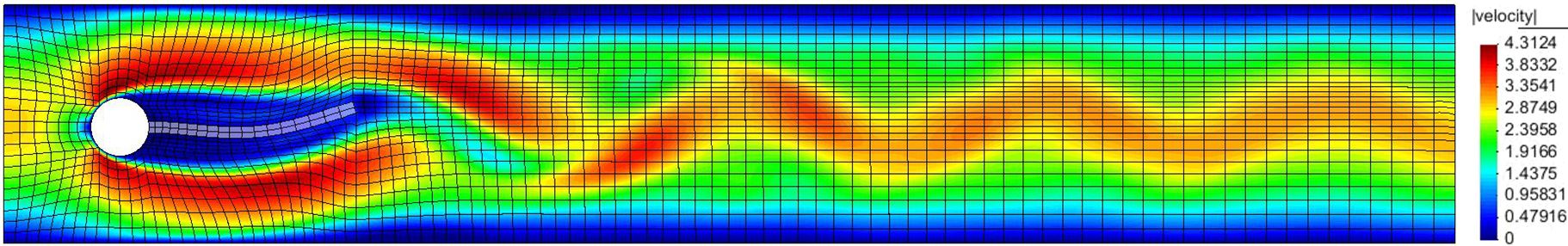}			
	\captionsetup{justification=centering}
	\caption {\scriptsize A snap shot of the velocity norms at t=6 using a coarse mesh.} 
	\label{flag}
\end{figure}

\begin{figure}[h!]
\centering  
\includegraphics[width=4.5in,angle=0]{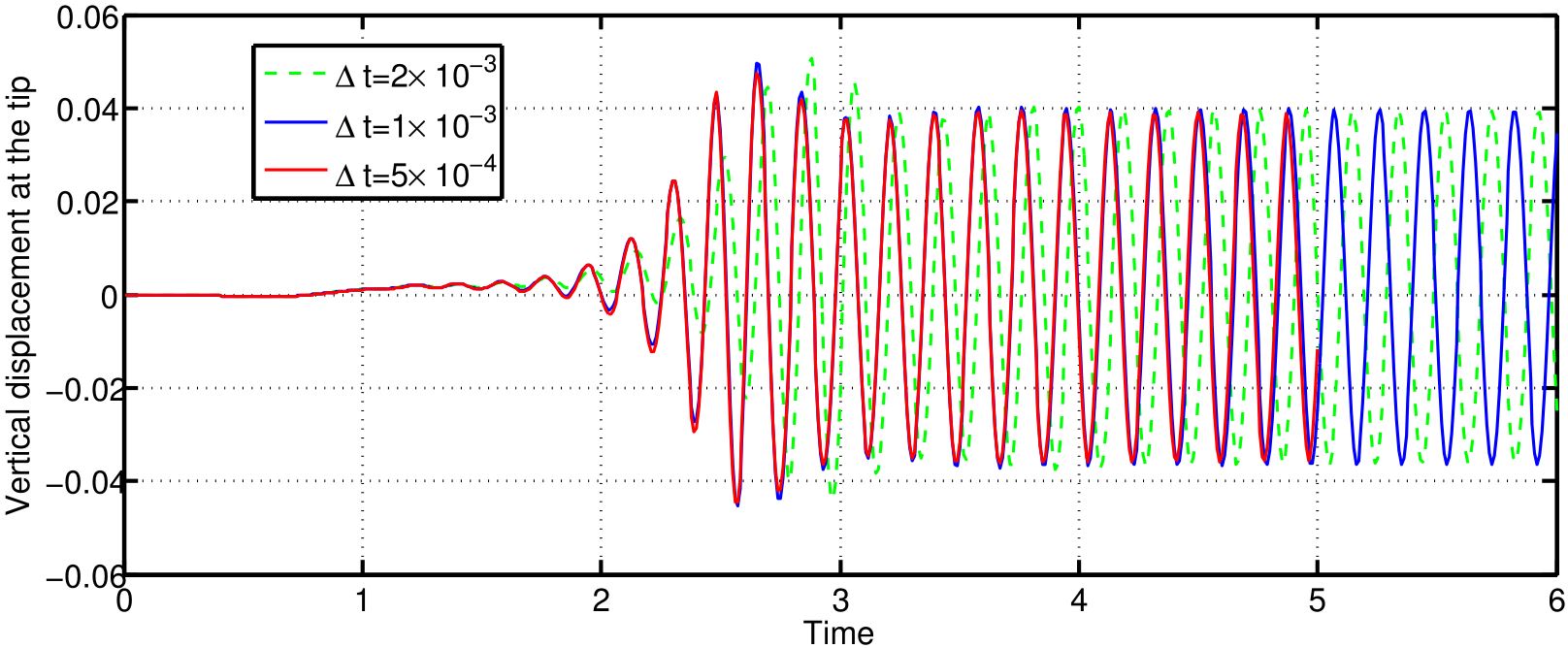}			
\captionsetup{justification=centering}
\caption {\scriptsize Vertical displacement at the flag tip as a function of time, using different time step and a medium mesh (data of the red curve is plotted up to $t=5$ for a better visualisation of the blue curve).} 
\label{flag_time}
\end{figure}

\begin{figure}[h!]
\centering  
\includegraphics[width=4.5in,angle=0]{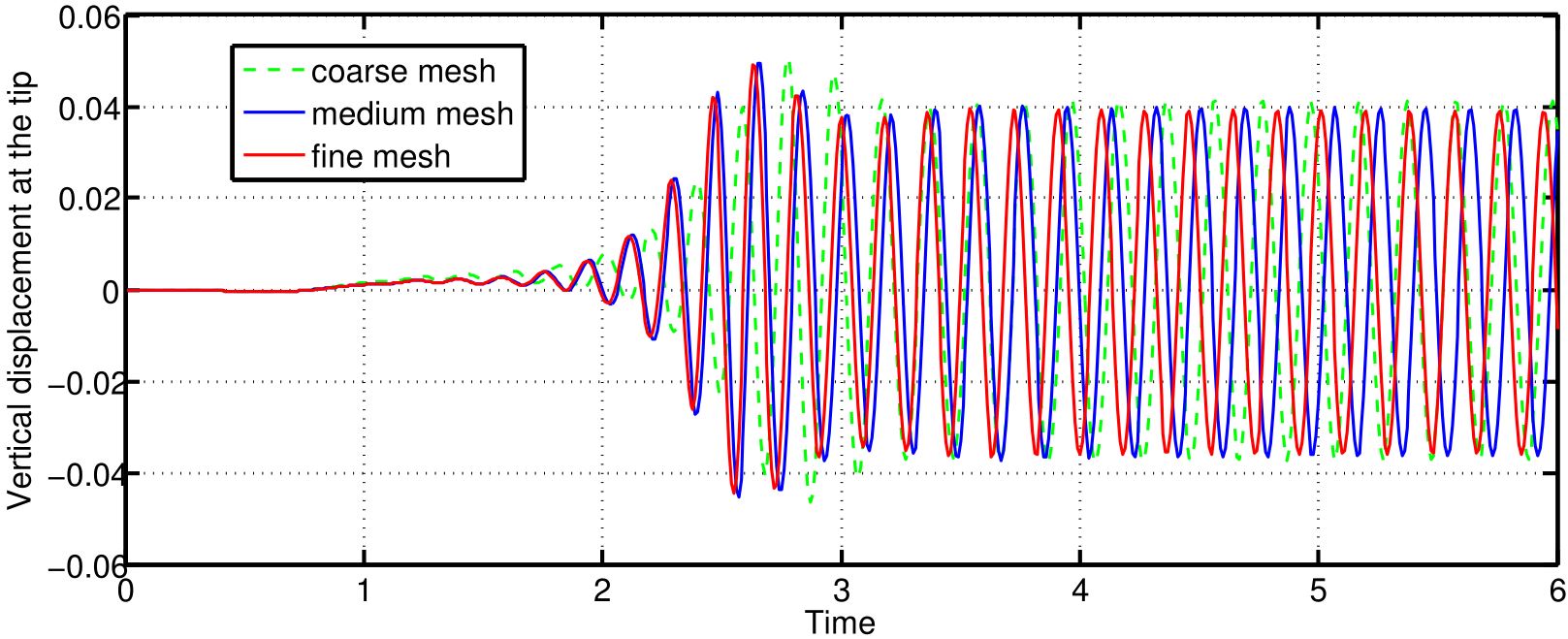}			
\captionsetup{justification=centering}
\caption {\scriptsize Vertical displacement at the flag tip as a function of time, using different mesh size and a time step size of $\Delta t=5\times 10^{-4}$.} 
\label{flag_mesh}
\end{figure}

\subsection{Falling disc}
\label{sec_fallingdisc}
In this test, we simulate a falling disc due to gravity \cite{Zhang_2007,Hesch_2014}, which needs remeshing in order to guarantee the mesh quality. However we will demonstrate that one needs much less remeshing, using the proposed ALE methods, compared to methods using pure remesh in order to fit the fluid-solid interface \cite{Hecht_2017}. This test is implemented using FreeFEM++ \cite{MR3043640}.

The computational domain is a vertical channel with a disc placed at the top of the channel as illustrated in Figure \ref{Computational_domain_falling}, where $W=4$, $H=12$, $h=2$ and $R=1$. In this test, $\rho^f=1$, $\rho^s=1.5$ $\mu^f=0.1$, $c_1=10^4$ and the gravity acceleration is $g=-9.81$. The fluid velocity is fixed to be 0 on all boundaries except the top. Notice that we choose $c_1$ sufficiently large so that the solid behaves as a rigid body. The computational domain is initially discretised by using 820 $P_2/P_1$ triangles with 1713 nodes as shown in Figure \ref{initial_mesh_faling}. We use a stable time step size of $\delta t=0.01$ and remesh every 100 times. We compare the simulation result against the empirical solution of a rigid ball falling in a viscous fluid \cite{hesch2014mortar}, for which the maximal velocity $U_m$ under gravity is given by
$$
U_m=\frac{\left(\rho^s-\rho^f\right)gR^2}{4\mu^f}\left[\ln\left(\frac{W}{2R}\right)-0.9157+1.7244\left(\frac{2R}{W}\right)^2-1.7302\left(\frac{2R}{W}\right)^4\right].
$$
In the test $U_m=1.2263$. The numerical and the empirical solutions agree well with each other when disc becomes stable as shown in Figure \ref{falling_v}. It can be understood that the disc velocity gradually decreases when it is close to the bottom of the channel. The evolution of the disc is displayed in Figure \ref{disc_evolution}. If we move the mesh by fluid velocity without the proposed ALE techniques, and remesh to guarantee the mesh quality. For this example, we find that remeshing has to be taken  at least every 7 time steps, otherwise the disc cannot successfully arrives at the bottom of the channel. We have also compared the {\bf F}-scheme and {\bf d}-scheme using this numerical test, and found that they presented very similar results although not showing in figure here.

\begin{figure}[h!]
	\begin{minipage}[h!]{0.5\linewidth}
		\centering  
		\includegraphics[width=1.2in,angle=0]{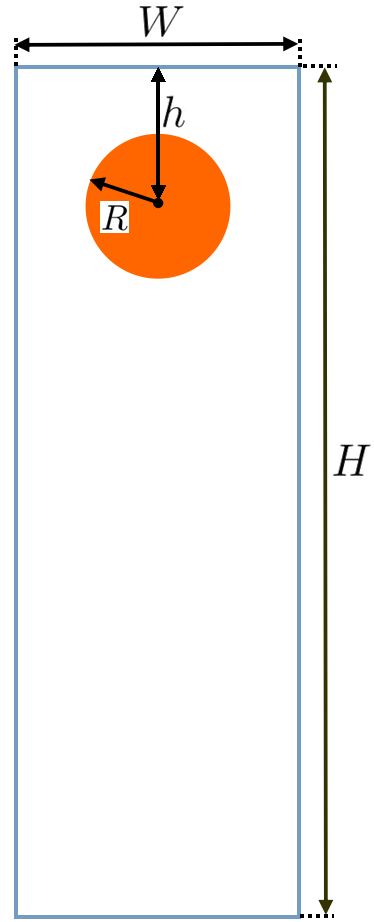}
		\captionsetup{justification=centering}
		\caption {\scriptsize Sketch of the falling disc.} 
		\label{Computational_domain_falling}	
	\end{minipage}
	\begin{minipage}[h!]{0.5\linewidth}
		\centering  
		\includegraphics[width=1.0in,angle=0]{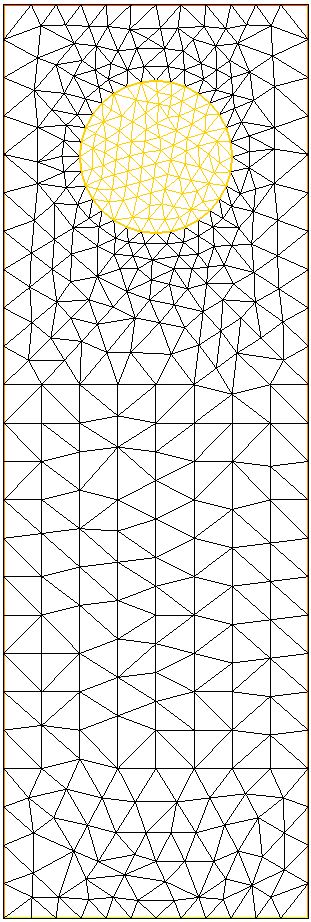}
		\captionsetup{justification=centering}
		\caption {\scriptsize Initial mesh for the falling disc.} 
		\label{initial_mesh_faling}		
	\end{minipage}     		
\end{figure}

\begin{figure}[h!]
	\centering  
	\includegraphics[width=3in,angle=0]{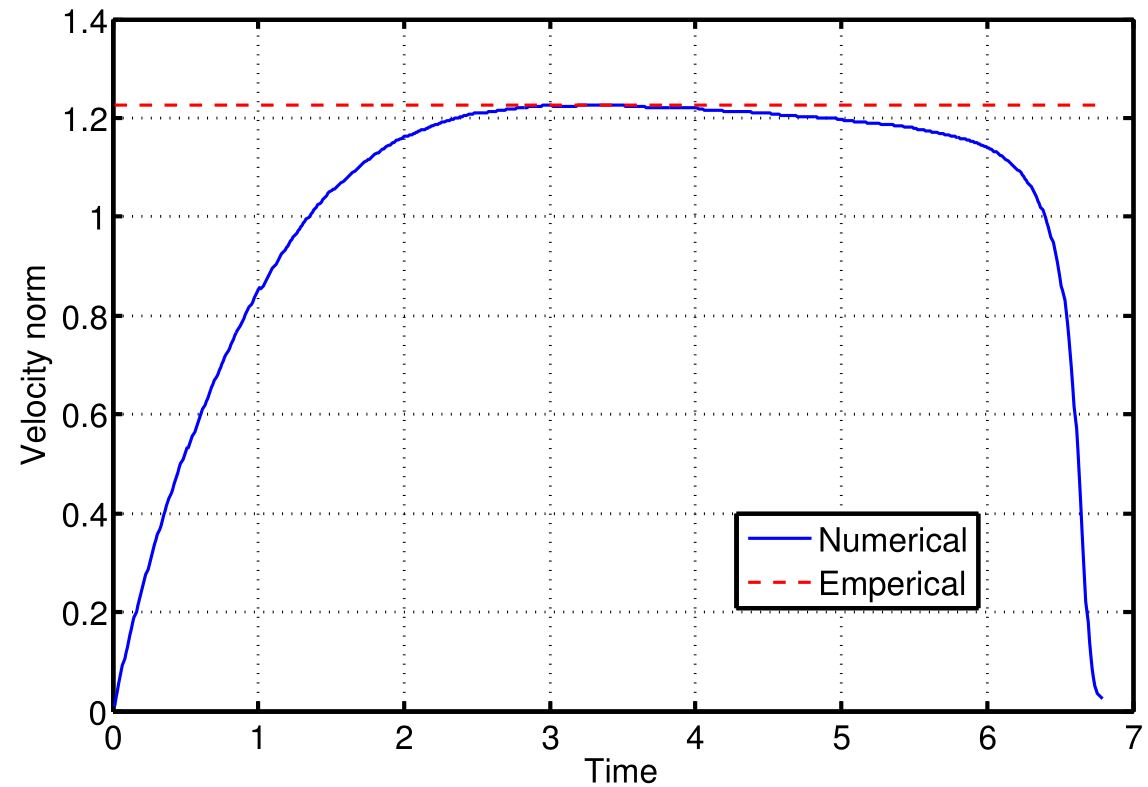}			
	\captionsetup{justification=centering}
	\caption {\scriptsize Comparison between the numerical and empirical velocity of the falling disc.} 
	\label{falling_v}
\end{figure}

\begin{figure}[h!]
\begin{minipage}[h!]{0.5\linewidth}
\centering  
\includegraphics[width=1.0in,angle=0]{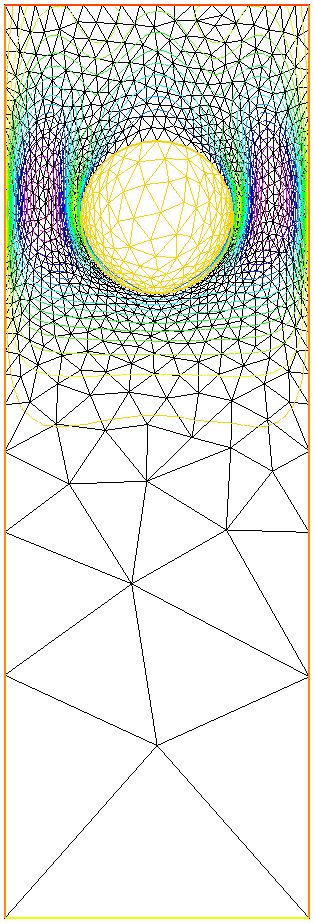}
\captionsetup{justification=centering}
\caption* {\scriptsize (a) $t=1.5$.} 	
\end{minipage}
\begin{minipage}[h!]{0.5\linewidth}
\centering  
\includegraphics[width=1.3in,angle=0]{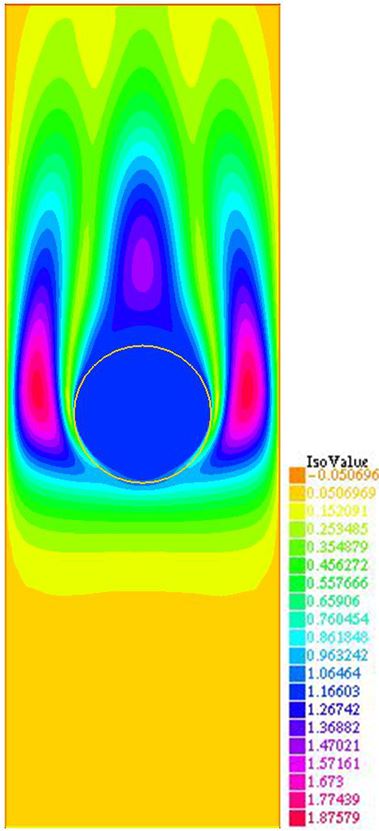}
\captionsetup{justification=centering}
\caption* {\scriptsize (b) $t=3.0$.} 
\end{minipage}   
\begin{minipage}[h!]{0.5\linewidth}
\centering  
\includegraphics[width=1.0in,angle=0]{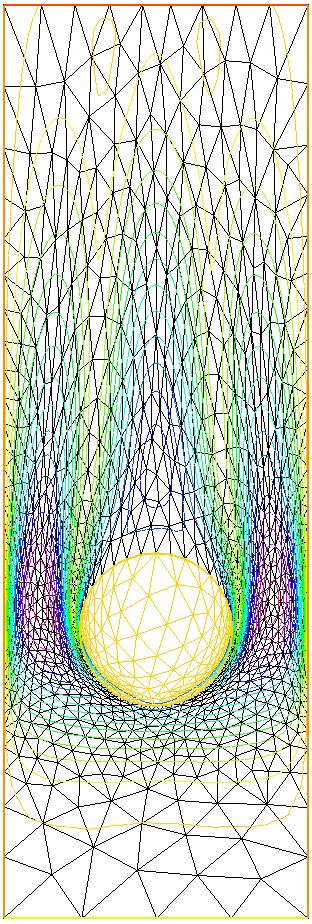}
\captionsetup{justification=centering}
\caption* {\scriptsize (c) $t=5.0$.} 
\end{minipage} 
\begin{minipage}[h!]{0.5\linewidth}
\centering  
\includegraphics[width=1.3in,angle=0]{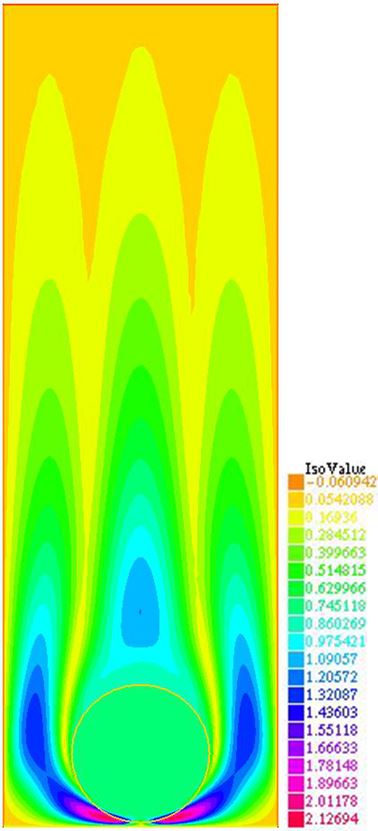}
\captionsetup{justification=centering}
\caption* {\scriptsize (d) $t=6.8$.} 
\end{minipage} 
\caption {\scriptsize Evolution of the falling disc, with colour showing the velocity norm.} 
\label{disc_evolution}	  		
\end{figure}

\section{Conclusion}
\label{sec_conclusion}
In this paper, we formulate the Fluid-Structure Interaction (FSI) system in an Arbitrary Lagrangian-Eulerian (ALE) coordinate system. The FSI system is formulated only using one-velocity field and solved in a fully-coupled manner. We prove this ALE-FSI formulation is unconditionally stable by analysing the total energy of the whole system. The stability result is achieved by expressing the problem in a conservative form, and adopting an exact quadrature rule in order to eliminate the mesh velocity. Several numerical tests are presented in order to validate the proposed scheme, including testing the energy stability, validating against a semi-analytical solution and a benchmark case, and combining with remeshing technique to simulate the case of extremely large solid displacement.

\bibliography{myref}
\end{document}